\documentclass[11pt]{article}
\usepackage[margin=1in]{geometry}
\usepackage{setspace}
\usepackage{amsmath,amssymb,amsthm}
\usepackage{graphicx}
\usepackage{hyperref}
\usepackage{algorithm}
\usepackage{algpseudocode}
\usepackage{refcount}
\usepackage{xcolor}
\usepackage{enumitem}
\usepackage{booktabs}
\usepackage{longtable}
\usepackage{array}
\usepackage{comment}
\usepackage{authblk}
\usepackage{lineno}
\usepackage{fancyhdr}

\usepackage[capitalize,nameinlink,noabbrev]{cleveref}
\crefformat{equation}{#2(#1)#3}
\crefmultiformat{equation}{(#2#1#3)}%
{ and~#2(#1)#3)}{, #2(#1)#3}{ and~#2(#1)#3}
\crefname{assumption}{Assumption}{Assumptions}
\crefname{problem}{Problem}{Problems}
\crefname{claim}{Claim}{Claims}
\crefname{fact}{Fact}{Facts}

\nolinenumbers

\theoremstyle{plain}
\newtheorem{theorem}{Theorem}
\newtheorem{lemma}{Lemma}
\newtheorem{corollary}{Corollary}

\theoremstyle{definition}

\theoremstyle{remark}

\definecolor{HeliaColor}{RGB}{255, 105, 180}
\newcommand{\Helia}[1]{\textcolor{HeliaColor}{Helia: #1}}
\definecolor{HeliaKColor}{RGB}{255, 0, 0}

\definecolor{grassgreen}{RGB}{0, 128, 0}

\definecolor{HedyehColor}{RGB}{255, 140, 0}
\newcommand{\hbcomment}[1]{\textcolor{HedyehColor}{Hedyeh: #1}}

\definecolor{MoColor}{RGB}{30, 144, 255}

\definecolor{MohammadColor}{RGB}{70, 70, 255}

\newcommand{\Cam}[1]{\textcolor{blue}{Cameron: #1}}

\usepackage{xspace}
\newcommand{\lowerCOSP}{0.262\xspace}
\newcommand{\lowerCOSPtau}{0.37\xspace}
\newcommand{\lowerCOSPbeta}{0.64\xspace}
\newcommand{\lowerCOSPgamma}{0.27\xspace}
\newcommand{\lowerCOSPdelta}{0.46\xspace}
\newcommand{\lowerCOSPtheta}{0.58\xspace}
\newcommand{\lowerROSP}{0.221\xspace}
\newcommand{\lowerROSPtau}{0.33\xspace}
\newcommand{\lowerROSPgamma}{0.34\xspace}
\newcommand{\lowerROSPdelta}{0.66\xspace}
\newcommand{\lowerROSPtheta}{0.63\xspace}

\newcommand{\ROSP}{ROSP\xspace}
\newcommand{\COSP}{COSP\xspace}

  \usepackage{nth}
  \usepackage{intcalc}

  \newcommand{\cAAAI}[1]{AAAI\ Conference\ on\ Artificial (AAAI)}

\title{The Secretary Problem with Predictions and a Chosen Order}
\author[1]{Helia Karisani}
\author[1]{Mohammadreza Daneshvaramoli}
\author[1]{Hedyeh Beyhaghi}
\author[1]{Mohammad Hajiesmaili}
\author[1]{Cameron Musco}

\affil[1]{University of Massachusetts Amherst\\
\texttt{\{hkarisani, mdaneshvaram, hbeyhaghi, hajiesmaili, cmusco\}@cs.umass.edu}}
\date{}

\begin{document}

\maketitle

\thispagestyle{empty} % Removes page number on the abstract page

% \section*{Possible Ideas}
% \begin{itemize}
%   \item prove choosing position of anything other than $\hat{i}$ doesn't give a better competitive ratio.
%   \item changing algorithm if there is little to no mistake, it will probably be no mistake. Using this fact and relay on prediction.
%   \item upper bound for COSP $\rightarrow$ \Dan{we found 0.33}.
%   \item upper bound for ROSP?
%   \item upper bound ideas that we had earlier using same formulas.
% \end{itemize}

\begin{abstract}
We study a learning-augmented variant of the secretary problem, recently introduced by Fujii and Yoshida (2023). In this variant, the decision-maker has access to machine-learned predictions of candidate values in advance. The key challenge is to balance  \emph{consistency} and \emph{robustness}: when the predictions are accurate, the algorithm should hire a near-best secretary; however, if they are inaccurate, the algorithm should still achieve a bounded competitive ratio.

We consider both the standard \textit{Random Order Secretary Problem (ROSP)}, where candidates arrive in a uniform random order, and a more natural model in the learning-augmented setting, where the decision-maker can choose the arrival order based on the predicted candidate values.  This model, which we call the \textit{Chosen Order Secretary Problem (COSP)}, can capture scenarios such as an interview schedule that is set by the decision-maker. 

We propose a novel algorithm that applies to both \ROSP and \COSP. Building on the approach of Fujii and Yoshida, our method switches from fully trusting predictions to a threshold-based rule when a large deviation of a prediction is observed. Importantly, unlike the algorithm of Fujii and Yoshida, our algorithm uses randomization as part of its decision logic.
We show that if $\epsilon \in [0,1]$ denotes the maximum multiplicative prediction error, then for \ROSP
our algorithm achieves competitive ratio
$
\max\left\{0.221,\; \frac{1-\epsilon}{1+\epsilon}\right\},
$
improving on a previous bound of $\max\left\{0.215,\; \frac{1-\epsilon}{1+\epsilon}\right\}$ due to Fujii and Yoshida~\cite{Fujii2023}.
For \COSP, our algorithm achieves
$
\max\left\{0.262,\; \frac{1-\epsilon}{1+\epsilon}\right\}.
$
This surpasses a $0.25$ upper bound on the worst-case competitive ratio that applies to the approach of Fujii and Yoshida, and gets closer to the classical secretary benchmark of $1/e \approx 0.368$, which is an upper bound for any algorithm. 
Our result for \COSP highlights the benefit of integrating predictions with arrival-order control in
online decision-making.

\end{abstract}

%\Cam{Make numbering start with 1 on the next page, not the abstact page. }

%\textbf{Keywords:} Secretary Problem, Algorithms with Prediction, Online Algorithms.

\clearpage
\setcounter{page}{1}
\section{Introduction}
\label{sec:intro}

%\Helia{Fix: classical setting and classical objective.}

The secretary problem, popularized by Gardner in his 1960 \textit{Scientific American} column~\cite{gardner1966new}, is one of the most well-studied problems in online decision-making. 
In the classic setting, candidates arrive in a uniform random order. When a candidate arrives, the  decision-maker observes a real-valued score for the candidate, and must make an immediate, irrevocable decision to either hire or reject the candidate. The classic objective is to maximize the probability of selecting the best candidate from the sequence. %\Cam{Add some sentences here about how the problem has been widely studied (with citations) and about how many variants have been considered (with citations and very brief descriptions). Some of this can be moved from related work.}
Alternatively, one can aim to maximize the \textit{expected value} of the selected candidate --  this is the variant we will focus on. The secretary problem and its many variants have been extensively studied \cite{Karlin2015OnAC,albers2019new,babaioff2007knapsack,bradac2020robust,buchbinder2014secretary}, and applied in domains such as resource allocation~\cite{babaioff2007knapsack} and mechanism design~\cite{kleinberg2005multiple}. %and active learning~\cite{sabato2017interactive}. 
%\hbcomment{Consider changing the last paper. It doesn't seem to be a seminal paper -- I think it has only 5 citations according to google scholar.}

Despite this work, the underlying assumption of the secretary problem -- that the decision-maker has no information about the values of candidates before they arrive -- is often unrealistic. In practice, decision-makers often have estimates of candidate values before interviewing them. For example, in applications to hiring, these predictions may be made based on resumes, portfolios, or public reputation. In fact, machine learning methods that leverage such information to make interview and hiring decisions are widespread in industry~\cite{Antoniadis2020SecretaryAO,shen2024impact,Raghavan2020,rieke2018help,chen2023ethics,MeBa2023,Mah2019,SaDah2023,ShaNe2021,MeBa2023}. 
%\Helia{I addresses Cameron's comment about citations on hiring with ML, as much as I could. I added some citations but kept the old ones too because they are still related and the reviewer might want to see them.}
%\Cam{Most of these refs don't look legit. They aren't in well known venues, don't have a lot of citations, and don't evidence that ML is commonly used for hiring in industry. By looking at \url{https://arxiv.org/pdf/2411.09854} and then looking through Google Scholar. I found several good refs. I left the urls for my refs commented out below.} 
%\url{https://scholar.google.com/scholar?hl=en&as_sdt=0%2C22&q=Help+wanted%3A+an+examination+of+hiring+algorithms&btnG=}, 
%\url{https://dl.acm.org/doi/abs/10.1145/3351095.3372828?casa_token=gdRDRSsWOn0AAAAA:9XnBiN00BwOpCFeDkPvluMAJVqSR5hJs73vutvdjxTsZFydyKHbO5CuJASXc0yOozNKizibUwM1lfA}, 
%\url{https://www.nature.com/articles/s41599-023-02079-x}} %\Cam{Except for the first, these citations don't look relevant. They are about theory papers on learning augemented online problems like caching. We need to look for papers that talk about the use for ML in hiring in industry. Probably more applied papers exploring e.g. the societal impacts of this. Perhaps look at the Antoniadis paper for ideas on citations}
%https://www.nature.com/articles/s41599-024-02647-9

%\Cam{We should cite all prior work on secretary with predictions since there isn't a lot of it.}
Given this, several recent works have studied variants of the Secretary Problem where the decision-maker has access to predictions, including estimates of the maximum value among the candidates~\cite{AntoniadisNIPS2020} or probabilistic forecasts of whether stronger candidates will appear later~\cite{benomar2023}.
In this work, we focus on a variant introduced by Fujii and Yoshida~\cite{Fujii2023}, where the decision-maker has prior access to a prediction $\hat{v}_i$ of the value $v_i$ of each candidate $i \in [n] := \{1, 2, \ldots, n\}$. They design an algorithm in this model that waits to select the candidate with the highest predicted value, unless any observed value deviates too much from its prediction; in that case, the algorithm switches to Dynkin’s classical worst-case \(1/e\)-competitive strategy~\cite{Dyn63}.

Let the \textit{competitive ratio} be defined as the worst-case ratio (over all possible assignments of candidate values) of the expected value of the candidate selected by an algorithm to the value of the true best candidate. Fujii and Yoshida~\cite{Fujii2023} prove that their algorithm achieves a competitive ratio of \(\max\left\{0.215, \frac{1 - \epsilon}{1 + \epsilon}\right\}\), where \(\epsilon = \max_{i} \left| 1 - \frac{\hat{v_i}}{v_i} \right|\) denotes the maximum multiplicative prediction error across all candidates. %\Cam{This was here but is out of place. Move this to later when it is used: "$\epsilon$ is the minium value such that, for all $i \in \{1,\ldots,n\}$, \((1 - \epsilon)v_i \leq \hat{v_i} \leq (1 + \epsilon)v_i\)."}\Dan{correct me if I am wrong but instead moving this out we should completely state Fuji's result then talk why it will be beneficial to fix beta and add gamma and beta} 
Their result  demonstrates that if predictions are accurate (i.e., \(\epsilon \to 0\)), the algorithm approaches optimal performance (i.e., the algorithm is \emph{consistent}); otherwise, it guarantees a worst-case constant-factor approximation (i.e., the algorithm is \emph{robust}). %In this work, we present our results in terms of the competitive ratio and our goal is to design algorithms that are both \textit{robust}—achieving a non-trivial competitive ratio even when predictions are inaccurate—and \textit{consistent}—approaching optimal performance when predictions are accurate.

\subsection{Our Contributions}

%\Helia{mention how our work differs from Fujii}

We build on the approach of Fujii and Yoshida, making the following main contributions:

\smallskip

\paragraph{Chosen Order Secretary Problem (\COSP) with Predictions.} Fujii and Yoshida restrict their attention to the classic setting in which candidates arrive in a uniform random order. We call this setting the  \textit{Random Order Secretary Problem (\ROSP)}. However, when predictions of candidate values are available up front, it is natural to use them to modify the arrival order of the candidates. For example, in applications to hiring, the interview order is typically set by the decision-maker and may be chosen based on the available predictions. To capture this setting, we introduce the \textit{Chosen Order Secretary Problem (\COSP)}, in which the decision-maker has full control of the candidate arrival sequence. \COSP is analogous to existing chosen order settings for prophet inequalities, in which a distribution over each candidate's value is known up front~\cite{hill1983prophet, bubna2023prophet,peng2022order}.
 However, like the classic secretary problem, \COSP is fundamentally worst-case, with no stochastic assumptions made on the inputs.

%In contrast to prior work, this paper introduces a more realistic version of the secretary problem by leveraging machine-learned predictions of the true value of each candidate and by choosing the arrival times of certain candidates, whereas most previous studies focus on random candidate sequences. This models a practical scenario where decision makers can schedule interviews based on their prior beliefs or stronger predictions. Throughout this paper, we refer to the randomized input model as the \textit{Random Order Secretary Problem} (ROSP), and to our new variant, where the arrival time of the candidate with the highest prediction is chosen by us, as the \textit{Chosen Order Secretary Problem} (COSP). Note that COSP setting is a modification on ROSP setting.

\smallskip

\paragraph{Improved Bounds for COSP.}
We design a simple algorithm for \COSP inspired by the approach of Fujii and Yoshida~\cite{Fujii2023}. Intuitively, one might imagine that interviewing candidates in e.g., decreasing order of their predicted values, would be optimal. However, it is easy to find simple worst-case instances that defeat these related strategies. Instead, our algorithm only uses the power of \COSP in a limited way: letting $\hat i \in [n]$ be the candidate with the highest predicted value, we place $\hat i$ at a fixed position in the arrival sequence, and randomize the order of all other candidates. We optimize the position of $\hat i$ to achieve maximum competitive ratio.  %We then use a randomized decision-making process to determine who to hire. 

After fixing the interview order, our algorithm proceeds similarly to that of Fujii and Yoshi. It initially  operates in a prediction-trusting mode, waiting to hire $\hat i$, the candidate with the highest predicted value. If it encounters a candidate whose true value significantly deviates from its prediction (i.e., exceeds a specified error threshold), it transitions to Dynkin’s classical threshold-based algorithm~\cite{Dyn63}. This algorithm ignores a fixed fraction of candidates, and then hires the first candidate (if any) it sees that is better than all the ignored candidates. Critically, we modify the approach of Fujii and Yoshida by not always hiring $\hat i$ if it is meant to be hired in the Dynkin mode of the algorithm. We introduce two parameters, $\gamma$ and $\delta$, that determine the probability of hiring $\hat i$ in different cases, and by tuning these parameters, are able to achieve a larger competitive ratio. %if we have switched into the Dynkin mode of the algorithm, and if $\hat i$ . %This randomization balances trust in predictions with robustness, helping avoid deterministic failure due to inaccurate predictions.

%\Cam{Giving the variable names iwthout any explaination of what they do is not useful to the reader. You need to explain at a very high level what the random decisions are.} This randomized mechanism improves robustness by preventing deterministic failures due to inaccurate predictions. Thus, while our method is structurally similar to Fujii and Yoshida’s, it generalizes their approach to adversarial arrival models \Cam{This is false.} and incorporates randomized fallback to strengthen performance guarantees. \Cam{Not sure what this last sentence means. What is 'randomized fall back'?}

We show that our algorithm is robust and consistent, achieving a competitive ratio of at least
\[
\max\left\{\lowerCOSP,\, \frac{1-\epsilon}{1+\epsilon}\right\},
\]
recalling that $\epsilon$ is the largest multiplicative prediction error. As in the algorithm of
Fujii and Yoshida, our method does not require $\epsilon$ as input. When $\epsilon \to 0$, the
competitive ratio approaches~1, and when $\epsilon$ is large, it remains bounded below by
$\lowerCOSP$. Importantly, our \COSP result surpasses the $0.25$ upper bound that arises in the
analysis of the Fujii--Yoshida algorithm in the random-order model.\footnote{Fujii and Yoshida
prove a $0.25$ upper bound for the random order setting where the decision-maker only sees the order of candidate arrivals (not their exact random arrival times), and where the decision algorithm is deterministic.
By leveraging arrival-order control and randomization, our algorithm exceeds this barrier. See
 \cref{sec:improved-bounds} for discussion.}

%We show that our algorithm is robust and consistent, %\Cam{Do not italicize terms after they have already been introduced before.} 
%achieving a competitive ratio of at least 
%\[
%\max\left\{\lowerCOSP, \frac{1-\epsilon}%{1+\epsilon} \right\},
%\]
%recalling that \(\epsilon\) is the largest multiplicative  error of any prediction. Notably, like the algorithm for Fujii and Yoshida, our method does not require \(\epsilon\) as an input. When \(\epsilon \to 0\), the competitive ratio approaches 1, and when \(\epsilon\) is large, it remains bounded by the constant \lowerCOSP.
 %\Cam{Say here that our algorithm overcomes a $.25$ upper bound that holds for the approach of Fujii and Yoshida. Mention in a footnote that they claim this upper bound is for all determinstic algorithm, but that this is not correct, and ref to an appendix for further discussion.}\Dan{This is not correct here as they don't claim for COSP but they claim for ROSP} \Cam{Well the lower bound holds for their algorithm either way, since it interviews in a random order. We need to mention it here like we do in the abstract. The point is that by using COSP we are able to overcome a natural earlier barrier for ROSP.}

\paragraph{Improved Bounds for ROSP.}
%
%We show that our same approach can be applied to the classic ROSP problem as well, achieving a competitive ratio of at least 
%\[
%\max\left\{ \lowerROSP,\; \frac{1 - \epsilon}{1 + \epsilon} \right\}.
%\]
%This improves on the previous best bound of 0.215 given by Fujii and Yoshida~\cite{Fujii2023}.
%As discussed, for COSP, our algorithm simply fixes the arrival time of the candidate with the highest predicted value, and randomizes the arrivals of all other candidates. Thus, to prove the above competitive ratio for ROSP, rather than optimizing the arrival time of the candidate with the highest prediction, we simply take the expectation %integrate
%over the random arrival time of this candidate. As in the COSP setting, our algorithm works without knowledge of \(\epsilon\). %\Helia{Where?: \Cam{Mention this when you initially mention the COSP guarantee.}}.
%
%\Cam{need to mention the .33 upper bound of \cite{choo2024} here and say that we make progress towards it but that the gap between upper and lower bounds for ROSP is still large}
We show that our approach also applies to the classic \ROSP setting, where instead of optimizing the arrival time of the best predicted candidate $\hat i$, we simply average over a uniform random choice of this time. In this case, by optimizing the parameters $\gamma$ and $\delta$, we can show that our algorithm achieves a competitive ratio of at least
\[
\max\left\{\lowerROSP,\; \frac{1-\epsilon}{1+\epsilon}\right\},
\]
improving on the previous bound of $\max\left\{0.215,\; \frac{1-\epsilon}{1+\epsilon}\right\}$ of Fujii and Yoshida~\cite{Fujii2023}. %In contrast
%to \COSP, the arrival time of the top-predicted candidate is now uniformly distributed in $[0,1]$.
%Accordingly, we analyze the algorithm by taking the expectation over this random arrival time and
%then optimizing the parameters of the secretary-mode rule. 
As in the \COSP setting, the algorithm
does not require~$\epsilon$ as input. Our worst-case competitive ratio gets closer to an upper bound of $0.33$, which \cite{choo2024} proves for any consistent algorithm solving \ROSP. However, a substantial gap between the best
known upper and lower bounds remains.

%Note that for \ROSP no
%algorithm can exceed the $0.33$ upper bound of~\cite{choo2024}.

%\hbcomment{Probably, we should not call this hardness. Also, need to rephrase this paragraph. But this could wait for later.}
% \paragraph{Hardness result for COSP.}
%We revisit the upper bounds proposed in~\cite{Fujii2023} and show that their claimed deterministic upper bound of 0.25 does not hold in general. We construct a counterexample—a randomized algorithm that achieves \(7/24 \approx 0.2917\)—which can be converted into a deterministic algorithm via internal randomness. This shows that the correct upper bound must be strictly greater than 0.25, contradicting their claim.

%More generally, we demonstrate that it is impossible for any algorithm to simultaneously match the optimal \(1/e \approx 0.367\) competitive ratio of the classical secretary algorithm and achieve a prediction-aware guarantee of the form \(1 - O(\epsilon)\). This formalizes the trade-off between robustness and consistency in this setting.

% \paragraph{New upper bound for ROSP.}
%On the upper bound side, we construct a new hardness instance and show that no algorithm can achieve a competitive ratio better than \upperROSP in ROSP, improving the prior upper bounds of 0.34 (from~\cite{Fujii2023}) and 0.33 (from~\cite{choo2024}). Our bound holds against all (possibly randomized) algorithms, assuming consistency when predictions are correct.

\subsection{Technical Overview}

We now overview the key ideas behind our improved algorithms for both \COSP and \ROSP. Throughout (see \Cref{sec:prelim} for details) we assume a continuous arrival model, in which each candidate $i \in [n]$ arrives at some time $t_i \in [0,1]$. In \ROSP, $t_i$ is chosen independently and uniformly at random. In \COSP, $t_i$ is set by the algorithm.

%\hbcomment{Mohammadreza: This section is nice, but a bit long.}

%\Dan{Part 1.}
A central requirement for consistency in the secretary problem with predictions is that the algorithm must hire the top-predicted candidate~$\hat{i}$ if they arrive before any candidate with a large deviation from its prediction has been observed. Otherwise, if all predictions end up being accurate (i.e., $\epsilon \approx 0$), the algorithm may not hire a candidate with the highest value ($v^* \approx \hat v$), and thus will not obtain a $\frac{1-\epsilon}{1+\epsilon} \approx 1$ competitive ratio. 
%Intuitively, if the predictions are accurate---that is, no candidate has deviated from their predicted value beyond the acceptable error threshold---then the candidate with the highest predicted value is likely to also be the true optimum. Hence, to achieve consistency, any algorithm must commit to selecting \(\hat{i}\) in such low-error settings.

The algorithm of Fujii and Yoshida~\cite{Fujii2023} builds on this principle. It begins in \emph{prediction mode}, where it plans to hire~$\hat{i}$ upon arrival, ignoring all preceding candidates. However, if the algorithm encounters a candidate whose observed value deviates significantly from its prediction -- a candidate we refer to as a ``mistake'' or ``prediction error'' -- before observing~$\hat{i}$, it abandons the prediction-guided strategy and switches to \emph{secretary mode}, where it applies the classical threshold-based algorithm of Dynkin~\cite{Dyn63}. This algorithm provides a worst-case $1/e$-competitive guarantee -- the best possible in the worst-case. %We use the term ``prediction error'' to refer to the event that such a mistake occurs.

The main technical challenge in the work of Fujii and Yoshida~\cite{Fujii2023} lies in analyzing the worst-case competitive ratio of their algorithm in the large $\epsilon$ regime. Once the algorithm switches to Dynkin’s strategy, the resulting behavior depends intricately on the timing and identity of the candidate that triggered the switch, as well as the arrival time of the remaining candidates. To establish robustness guarantees, they perform a detailed case analysis that accounts for all possible positions of the top predicted and the true top candidate, as well as the arrival of all candidates with high prediction error. %\Cam{What does it mean the configuration of the prediction errors? Do you just mean the number of prediction errors? Be more precise.}
Ultimately, they show that in the large $\epsilon$ setting, their algorithm achieves a competitive ratio of $0.215$. Note that a natural upper bound here is  \(1/e \approx 0.368\) -- which would be obtained if one knew in advance that 
%
%---------------------------------------------------
%
%
%\begin{itemize}
%\item Explain that any algorithm which is \emph{consistent} needs to hire $\hat i$ if it shows up when you haven't seen many mistakes.
%
%
%
%\item Fujii and Yoshida start with this fact and there algorithm is in Prediction node (i.e., waiting to hire $\hat i$) until it sees a mistake. If it sees a mistake it switches to the standard worst-case algo (Dynkin).
%\item Challenge of their paper is doing a complex case analysis to analyze the worst-case competitive ratio.
%\end{itemize}
%
%\Dan{Part 2.}
there would be a significant prediction error, and followed Dynkin’s rule from the beginning.
Moreover, in the random-order model, any algorithm that is consistent—i.e., achieves a competitive ratio approaching $1$ when predictions are accurate—cannot have robustness exceeding $0.33$; this follows from the upper bound of~\cite{choo2024} for ROSP.
This bound applies to consistent algorithms operating under random arrival, but whether an analogous limitation holds for \COSP\ is currently unknown.

%, which achieves a competitive ratio of \(1/e \approx 0.368\)—the best possible in the worst case. This sets a natural robustness benchmark.
%\Cam{Don't we know that the upper bound of $0.33$ from \cite{choo2024} for $\ROSP$ also applies to \COSP? If so this should be mentioned here, and also earlier when we first state our result for \COSP and compare it to the $.25$ upper bound of Fujii and Yoshida.}

\paragraph{Fixing the arrival of $\hat i$.}
Interestingly, in the analysis by Fujii and Yoshida, the arrival position of the top-predicted candidate~$\hat{i}$ has a significant impact on the competitive ratio. This observation motivates our focus on \COSP, and on strategically fixing the arrival time of $\hat i$ in particular. % that, rather than assuming a uniformly random arrival time for~$\hat{i}$ as in their study, we can strategically fix its position to improve performance.
To illustrate this, suppose we fix $\hat i$'s arrival time to \(\beta \approx 0\) and keep all other arrival times uniformly random. The adversary can set all predictions correctly except for the true best candidate $i^*$, which is underestimated. With probability \(\beta\), this candidate arrives before \(\hat{i}\), triggers a switch to secretary mode, and is skipped. With probability \(1 - \beta\), the algorithm hires \(\hat{i}\) before ever seeing $i^*$. In either case, the competitive ratio approaches zero. On the other hand, if \(\beta \approx 1\), the adversary may overestimate \(\hat{i}\), causing the algorithm to skip all earlier candidates in favor of \(\hat{i}\), which results in success only with probability \(1 - \beta\) -- again close to zero. Therefore, setting \(\beta = 0\) or \(\beta = 1\) is provably suboptimal, and we instead optimize over \(\beta\) to find a value that balances these trade-offs and maximizes the competitive ratio.
%
%\hbcomment{Mohammadreza: maybe the next paragraph is a potential candidate for deleting / shortening.\\}
%This approach is also motivated by practical considerations: in many real-world scenarios such as hiring, the decision-maker often has control over the interview schedule and may prefer to place high-confidence candidates earlier or later in the sequence based on prediction quality and risk preferences.
We note that placing a distribution over the arrival time \(\beta\) of \(\hat{i}\), rather than fixing it deterministically, may further improve performance. However, we leave this possibility for future work. % and focus in this paper on optimizing a fixed \(\beta\). 
We conjecture that controlling the arrival times of candidates other than $\hat i$ (rather than simply randomizing them) cannot significantly improve our bounds -- however, we again leave investigating this question for future work.
\begin{comment}
    
\begin{itemize}
\item If you knew there was going to be a mistake, should just randomize the order and achieve $1/e$ competitive, which is the best possible in general. So fixing the order it really to help in the case with mistakes. So it's natural to just fix the position of $\hat i$, who must be hired if there are no mistakes.
\item Can't fix $\beta = 0$. Can't fix $\beta = 1$. So we'll optimize over $\beta$ and find some $\beta$ value giving the best competitive ratio.

\item Minor: Open if we can put a distribution on $\beta$ and do better but leave that for later.
\end{itemize}
\end{comment}

%\Dan{Part 3.}

%\hbcomment{$\tau$ has not been defined yet.\\}

%\Cam{I've gotten to here so far.}

\smallskip

\paragraph{Tuning the probability of hiring $\hat i$.}
Unfortunately, simply optimizing the arrival time of $\hat i$ and applying the algorithm of Fujii and Yoshida~\cite{Fujii2023} directly can fail. Recall that Dynkin's worst-case algorithm ignores all candidates before some threshold time $\tau$, and then hires the first candidate that arrives after $\tau$ and is better than all preceding candidates. If we set \(\beta < \tau\) and the top predicted candidate \(\hat{i}\) is in fact the true best candidate, but its prediction deviates significantly from its value, the algorithm switches to secretary mode upon its arrival and rejects \(\hat{i}\) since it arrives before $\tau$. This can result in a near zero competitive ratio. Conversely, if we fix \(\beta > \tau\), consider a scenario with no deviation between the predictions and the true values, except that \(\hat{i}\) is predicted to be the best candidate but is actually the second-best candidate (and by a wide margin). If the true best candidate $i^*$ arrives before \(\beta\), it will be ignored since the algorithm remains in prediction mode. If it arrives after \(\beta\), the algorithm will have already selected \(\hat{i}\) as it would be the best candidate seen after $\tau$ and after switching to secretary mode. Again, we have a competitive ratio of near zero.

The handle the above issues, we set $\beta > \tau$, but modify the algorithm to not always hire \(\hat{i}\) when it is the first mistake and the best candidate observed so far. Instead, we introduce a parameter \(\gamma\) that specifies the probability of hiring \(\hat{i}\) in this case. This change ensures that the algorithm retains flexibility in handling situations where blindly trusting \(\hat{i}\) could lead to poor outcomes. Optimizing \(\gamma\) allows us to avoid deterministic failure while still preserving good performance in other cases.

We introduce a second parameter \(\delta\) to handle the case where \(\hat{i}\) is not the first observed mistake. It does not make sense to use the same hiring probability \(\gamma\) in this setting, as the context is different: the algorithm has already seen one or more mistakes and switched to secretary mode before seeing $\hat i$. The likelihood that \(\hat{i}\) is truly optimal may be lower or higher than the case when it is the first mistake %\Cam{Unsure what this means. Likelihood under what distribution? Lower or higher than why?}
, so we allow a separate hiring probability \(\delta\) This additional degree of freedom enables us to optimize the algorithm's behavior even further.

%\Dan{Part 4.}

As in the work of Fujii and Yoshida, a key technical challenge is to analyze the competitive ratio
against all possible adversarial inputs. This requires accounting for every configuration of the true
best candidate $i^*$, the top-predicted candidate $\hat{i}$, and the locations of all large-deviation
candidates. As in their analysis, this leads to a detailed case decomposition. Following their
notation, let $M$ be the set of large-deviation candidates, and let $m = |M|$. We further refine the structure of the high-deviation set $M$ by comparing the
true values of its members to the true value of the best–predicted candidate~$\hat{i}$.
Specifically, we define
\[
m_2 \;=\; \bigl|\{\, i \in M : v_i < v_{\hat{i}} \,\}\bigr|,
\qquad
m_1 \;=\; m - m_2,
\]
so $m_2$ counts the number of ``mistakes’’ whose true value is \emph{smaller} than
the true value of~$\hat{i}$.  
Although both quantities are well defined, only $m_2$ plays a role in our
competitive-ratio bounds.
In addition, we define
\[
k \;=\; \bigl|\{\, i \in [n] : v_i > v_{\hat{i}} \,\}\bigr|,
\]
which is the number of candidates whose true value exceeds that of the
best–predicted candidate~$\hat{i}$.  
The parameter $k$ captures how many
truly better candidates the adversary places relative to~$\hat{i}$, and therefore
directly influences the worst-case behavior of the algorithm in prediction mode.

These parameters allow us to
partition the analysis into fine-grained subcases, each capturing a distinct adversarial structure.
We then optimize over $\beta$, $\gamma$, and $\delta$ to obtain the worst-case competitive ratio.

Despite the added generality, the algorithm retains a simple structure amenable to analysis. By carefully organizing the cases based on the parameters, we are able to compute the worst-case competitive ratio numerically. Through this analysis, we obtain a provable lower bound of $\max\left\{\lowerCOSP, \frac{1-\epsilon}{1+\epsilon} \right\}$ for \COSP, demonstrating the benefit of tuning \(\beta\), \(\gamma\), and \(\delta\).

\paragraph{Extension to \ROSP.}
In the \COSP\ setting, for each fixed value of $\beta$, we evaluate the competitive ratio across all
adversarial configurations and then optimize the parameters $\gamma$ and $\delta$ for that fixed
$\beta$. We then optimize over $\beta$ to obtain the final \COSP\ bound.

In the \ROSP\ setting, however, the arrival time of $\hat{i}$ is itself a uniform random variable in
$[0,1]$, and cannot be chosen by the algorithm. Accordingly, the analysis must first take the
expectation of the competitive ratio over $\beta \sim \mathrm{Unif}[0,1]$. We then optimize the
parameters $\gamma$ and $\delta$ with respect to this new competitive ratio objective. This yields a lower bound
of
\[
\max\left\{ \lowerROSP,\; \frac{1 - \epsilon}{1 + \epsilon} \right\},
\]
slightly improving on the prior bound of~\cite{Fujii2023}.
%\Cam{The above description doesn't sound right. Since you can't optimize for each $\beta$ and then take expectation, since you don't know the arrival time of $\hat i$ in advance. You have to take the expectation over $\beta$ and then optimize $\gamma$ and $\delta$. This is what we did right? Need to reword to explain this.}

\paragraph{Why only choose arrival of $\hat{i}$ and the decision at arrival time?}
One natural question is whether introducing additional parameters—for
example, tuning the algorithm's behavior when encountering the
second- or third-highest predicted candidates—could further improve the
robustness.  Our analysis, as well as the structure of~\cite{Fujii2023},
suggests that this is unlikely: in all adversarial instances that
determine the worst-case competitive ratio, the only critical decision
point is when the algorithm reaches $\hat{i}$ or first mistake(deviated prediction) while still following the
classical secretary rule.  Lower-predicted candidates never influence these worst-case configurations in a meaningful way unless they are mistakes (deviated predictions).  By contrast, in extensions of the
secretary problem where the goal is to select the top $k>1$ individuals or
where multiple high-valued candidates may be accepted, controlling the
arrival times of several top-predicted candidates could indeed be
beneficial.  Developing such multi-candidate scheduling strategies is an interesting direction for future work.

\subsection{Related works}

\smallskip

\paragraph{Classical and Variant Secretary Problems.}
%\Cam{We said earlier it was popularized in 1960 by Garfield. Who actually introduced it?} \Helia{It was popularized by Garfield but later solved by Dynkin, I think I think fix belos lines.}
The classical secretary problem, popularized by Gardner~\cite{gardner1966new} and others~\cite{Dyn63,gilbert1966recognizing,lindley1961dynamic,esfandiari2020prophet,Liu2021QueryBasedSO,kaplan2020competitive,correa2021secretary,Salem2023,Correa2021,Azar2018}, has inspired a broad and influential line of work in online algorithms, economics, and decision theory. Numerous extensions have been studied, including multiple-choice variants~\cite{kleinberg2005multiple}, matroid and knapsack constraints~\cite{babaioff2007matroids}, online matching and graphic/transversal matroids~\cite{korula2009algorithms}, submodular objectives under cardinality and packing constraints~\cite{kesselheim2016submodular}, and structured feasibility systems such as packing integer programs~\cite{argue2021robust}. We refer the reader to the survey by Gupta and Singla~\cite{gupta2020randomorder} for a broader overview of results in random-order models. %\hbcomment{feel free to add more references. happy to add more on specific topics.}

%Since its introduction in the 1960s~\cite{gardner1966new}, the secretary problem has been extensively studied~\cite{gilbert1966recognizing,lindley1961dynamic}. The well-known $(1/e)$-competitive algorithm~\cite{Dyn63} is known to be optimal for the classical formulation. A vast body of work explores variants of this problem, including multiple-choice secretary problems~\cite{kleinberg2005multiple}, settings with non-uniform arrivals~\cite{Liu2021QueryBasedSO}, and game-theoretic extensions~\cite{Karlin2015OnAC}. Several works~\cite{campbell1981choosing,kaplan2020competitive,correa2021secretary} consider models where a sample of candidates drawn from the same distribution is revealed in advance to assist the selection process. Other lines of work study partially adversarial arrivals~\cite{kesselheim2020knapsack,bradac2020robust}, where adversarial and random-order candidates are mixed.

\smallskip

\paragraph{Secretary Problems with Predictions.}
Recent work incorporates predictions into the secretary problem to enable improved performance. In the model that we adopt, each candidate comes with a predicted value $\hat{v}_i$ estimating their true value $v_i$, as in~\cite{Fujii2023,choo2024}. Other forms of predictive feedback include binary comparisons~\cite{benomar2023}, probabilistic signals about candidate quality~\cite{Moustakides2023,dutting2021secretaries}, single-shot predictions of the best candidate~\cite{AntoniadisNIPS2020}, and predictions of the value difference between the best and the k-th best candidate~\cite{braun2024predicted}. Some work has studied algorithms that ensure fairness by guaranteeing the best candidate is accepted with constant probability, even under biased predictions~\cite{balkanski2024fair}.
Additionally, Amanatidis et al.~\cite{amanatidis2025online} study value maximization with predictions in an online budget-feasible mechanism design setting, using techniques related to the secretary problem. %\Cam{We are missing \cite{balkanski2024fair} in the above paragraph. Also its a bit non-paralell}
%{we just talk about the models for most citations but then talk about the results for Braun and Sarkar. Why do we focus more on that paper?}

\paragraph{Prophet Inequalities and Order Selection.}
Prophet inequalities are a classical model in optimal stopping theory that can be seen as a stochastic version of the secretary problem. A decision-maker observes a sequence of values drawn from known distributions and aims to compete with a benchmark that selects the maximum in hindsight. The classical result guarantees a $1/2$ competitive ratio~\cite{krengel1978semiamarts,samuel1984comparison}, which has inspired a rich body of work in theoretical computer science and economics; see surveys such as~\cite{hill1992survey,lucier2017economic,correa2019recent}. Numerous variants have been studied, including the adversarial order model~\cite{krengel1978semiamarts}, the random order (prophet secretary) model~\cite{esfandiari2015prophet}, and the order-selection (free-order) model~\cite{hill1983prophet}. This model is especially related to our Chosen Order Secretary Problem (COSP), in which the arrival order is determined by predicted candidate values. Relevant works in the order-selection model include~\cite{chawla2010multi,beyhaghi2021improved,peng2022order,bubna2023prophet}.
Intuitively, prophet inequalities capture similar scenarios in which the decision-maker has in advance information about candidate values; however, they operate under stochastic assumptions (values drawn from known distributions), whereas our setting is fully worst-case, with no distributional assumptions. 
Interestingly, until recently, it was unclear whether the order-selection and random-arrival models for profit inequalities differed in their achievable competitive ratios. A separation has now been established for prophet inequalities~\cite{bubna2023prophet}, but in our setting, the relationship remains open: the best upper bound for the random-arrival variant (\ROSP) currently exceeds the best known competitive ratio for order selection (\COSP).
%\Helia{Hedyeh addressed this comment above.}\Cam{I think looks good. I think we should explicitly say that profit inequalities can be used to model the same type of situation as our problem: i.e., when the decision-maker has some information about the values of the candidates before they arrive. But that profit inequalities are stochastic while our problem is worst-case and doesn't make any distributional assumptions.} \hbcomment{Let me know if this addresses your comment---feel free to edit. I also added a short part about the separation of random-order and free-order.}

\paragraph{Learning-Augmented and Online Algorithms with Predictions.}
A growing literature studies learning-augmented online algorithms in which predictions guide online decisions without compromising worst-case robustness. Notable examples include caching~\cite{Lykouris2018,Rohatgi2020,Wei2020approx}, ski rental~\cite{Purohit2018,Gollapudi2019,WeiZhang2020,diakonikolas2021distributional}, online scheduling~\cite{Lattanzi2020}, and metrical task systems~\cite{Antoniadis2020}. Learning-augmented approaches have also been applied in online primal-dual methods~\cite{Bamas2020b}, speed scaling~\cite{Bamas2020a}, online bidding~\cite{aggarwal2025autobidding}, and auctions~\cite{gkatzelis2024clock}. We refer the reader to surveys such as~\cite{Mitzenmacher2020} for comprehensive overviews. Moreover, learning-augmented approaches have also been applied to online purchasing-type problems with prediction, such as one-way trading, online conversion~\cite{alBin97,Sun21,Lec24}, and online knapsack problems with prediction~\cite{ImRavi2021,gehnen2024,Boyar24,Daneshvaramoli2025}. These problems often share structural similarities with secretary problems—e.g., online knapsack with unit-size items and unit capacity resembles the secretary problem under random arrival.

%\Helia{working on it.}
%\Cam{Are there any problems that are specifically relevant to us? Seems like 1-way trading, knapsack, and other purchasing type problems are very relevant so should be more the focus of the citations? Like is secretary a special case of online knapscak where all items have unit size and the knapsack has unit capacity? Thus e.g. our own prior work on online knapsack would apply.}

%%%%COMMENTED OUT BY IFFALSE COMMAND%%%%%
\iftrue

\subsection{Outline of the Paper}

%\Helia{should be shorter at the end}
\begin{comment}
    The remainder of this paper is organized as follows. Section 2 formally defines our predictive secretary model and reviews relevant literature. Section 3 introduces our main algorithm, providing detailed theoretical analyses and guarantees. Section 4 presents extensive experimental results illustrating the empirical advantages of our method compared to baseline and state-of-the-art algorithms. Finally, Section 5 summarizes our findings and highlights potential directions for future research.
\end{comment}

%\hbcomment{Mohammadreza: This part needs to be finalized.}
%\Helia{It has to be done in the end.}

This paper is organized as follows. In \cref{sec:prelim}, we give preliminaries and formally define our problem settings and the notation. In \cref{sec:tbgd-method}, we study the Chosen Order Secretary Problem (\COSP), presenting a randomized algorithm that achieves competitive ratio $\max\{\lowerCOSP, \frac{1 - \epsilon}{1 + \epsilon}\}$. In \cref{sec:rand}, we study the Random Order Secretary Problem (\ROSP), extending our algorithm for \COSP to obtain a competitive ratio of $\max\{\lowerROSP, \frac{1 - \epsilon}{1 + \epsilon}\}$ in this setting. Finally, in \cref{sec:improved-bounds}, we examine the relationship between deterministic and randomized algorithms in the classic secretary problem with predictions, showing that under the standard value-maximization formulation, both models achieve the same competitive ratio.

\section{Preliminaries}
\label{sec:prelim}

We let \( N= \{1, \ldots, n\} \) denote the set of candidates. %\Cam{Need to explain that for ROSP the arrival times are $unif(0,1)$ but for COSP they are set by the decision-maker.} 
Throughout, we adopt a continuous-time framework, where each candidate~$i$ is assigned an arrival time~$t_i \in [0,1]$. 
In \ROSP, each $t_i$ is drawn independently from \(\mathrm{Unif}(0,1) \)~\cite{kleinberg2005multiple,bruss1984unified,feldman2011improved}. This induces an arrival order that is a uniform random permutation. In \COSP, the decision-maker selects the arrival time $t_i$ of each candidate. %\Cam{This was not mentioned above. You need to first say that the times are continuous in $0$ to $1$. Then you can say how they are set.} 
Note that in \ROSP, for any subinterval~$A \subseteq [0,1]$, the probability that candidate~$i$ arrives within~$A$ is equal to its length. %, i.e.,~$\Pr[t_i \in A] = \text{length}(A)$. In this setting,~$t_i$ acts as a randomized proxy for the candidate’s relative arrival position. Therefore, each candidate \( i \) arrives over the interval \([0,1]\).

%Sorting the vector \((t_1,\dots,t_n)\) yields the arrival order, with earlier times indicating earlier arrivals. 
Because all arrival times are distinct with probability 1, the continuous-time model is equivalent to the related model where candidates are just assigned some arrival position in $[n]$ -- in that model, the decision-maker could always draw a set of random arrival times to assign to the candidates based on their arrival order. 

Each candidate \( i \in [n] \) has a true value \( v_i \in \mathbb{R}_{\geq 0} \), which remains unknown to the decision-maker until the candidate appears at time \( t_i \). Upon observing each candidate, the decision-maker must irrevocably choose whether to select or reject them; once a candidate is selected, the process terminates.
We evaluate performance under the value-maximization objective: maximizing the expected value of the selected candidate, where the expectation is taken over any randomness in the arrival order and any randomness used by the selection algorithm. As defined before, the performance measurement tool is competitive ratio, the worst-case ratio (over all possible assignments of candidate values) of the expected reward of the algorithm to the maximum true value, i.e., \( \mathbb{E}[\text{ALG}] / v^* \), %\Cam{Don't know what the previous equations means. Also Don't think we need to define $r$. Just say its the ratio. And then define both $v^*$ and $i^*$ in the paragraph afterwards.} 
where $\mathbb E[\text{ALG}]$ denotes the expected reward achieved by the algorithm and \( v^* = \max_{i \in [n]} v_i \) is the highest true value among all candidates, corresponding to candidate $i^*$ arrived at time $t^*$. %\Cam{I couldn't follow previous sentence. In particular $r/v*$ being the expected reward?}

 In addition to the true values \( v_i \), a prediction \( \hat{v_i} \in \mathbb{R}_{\geq 0} \) is provided in advance for each candidate \( i \in [n] \), representing an estimate of \( v_i \).  We let \(\hat{i}\) denote the candidate with the highest predicted value arriving at time \( \hat{t} \). Let \( \epsilon = \max_{i \in [n]} \left| 1 - \frac{\hat{v}_i}{v_i} \right| \) denote the largest 
multiplicative prediction error among all candidates. We will often make use of the following error guarantee on predictions:

\begin{equation}
    \label{eq:pred-bound}
    (1 - \epsilon)v_i \leq \hat{v_i} \leq (1 + \epsilon)v_i \quad \text{for all } i \in [n].
\end{equation}

%\Cam{Are we assuming a unique maximum? If so should say this is ok because we can use some tiebreaking rule?} \Dan{use this tie breaker rule $v_i$,$v'_i$ = actual $v_i$ + candidate number/$10^{100000000}$ no tie and you will get +- $10^{100000000}$ error} \Cam{Add explination of this. Don't use concrete numbers. Just replace $1/10^{100000000}$ with some $\delta$ that can be set arbitarily close to $0$}

%\hbcomment{Mohammadreza: In the next paragraph and the formula that comes right after, shouldn't we use $\hat{v}$? Also, ``old value'' in the formula is a bit weird.}
Without loss of generality, we assume uniqueness for both the predicted values and the true values; hence, there is a uniquely defined top candidate $i^*$ and top predicted candidate $\hat i$. 
To handle cases with repeated values, we can slightly perturb the values to break ties. In particular, when the same true value $v_i$ appears more than once in the sequence, we treat each subsequent occurrence as slightly larger than the previous one by adding an infinitesimal offset (e.g., \(v \mapsto v + \eta\) with \(\eta > 0\) arbitrarily small). This guarantees that all values are distinct, while ensuring that the modified values remain arbitrarily close to the original ones. The algorithm is then competitive with respect to this perturbed sequence, and since the perturbed values differ from the true values by at most an infinitesimal factor, the resulting competitiveness carries over to the original sequence. We can use a similar approach for the predicted values $\hat v_i$.

\section{Chosen Order Secretary Problem (COSP) with Prediction}
\label{sec:tbgd-method}

In this section, we analyze the competitive ratio of our algorithm. We begin by providing intuition for the algorithm, then define the model and set up the key parameters. We then analyze the competitive ratio under a fixed arrival time \(\beta\) of the top predicted candidate.

%\Helia{I am waiting for the comments exchanged below to be resolved then combine the 3.1 and 3.2. Cameron's comment to be applied: combining into one streamlined section that basically explains the idea behind the algorithm and the intuition behind why we do different things (i.e., what it is helpful to fix beta and why delta and gamma are helpful)}
\subsection{Algorithm Intuition}

%\Helia{Will do after I combine 3.1 and 3.2} 
%\Cam{You always need at least a sentence before the first subsection. A natural thing to do is to summarize what you plan to do in each subsection of this section.}

We describe our algorithm for the COSP setting, where the arrival order is partially controlled based on predictions. Specifically, the candidate with the highest predicted value---denoted by $\hat{i} = \arg\max_i \hat{v}_i$---is scheduled to arrive at a fixed time $\beta \in [0,1]$, i.e., we set $\hat{t} = \beta$, while all other candidates are assigned independent arrival times drawn uniformly from $[0,1]$.

The algorithm operates in two modes: \textbf{Prediction mode} and \textbf{Secretary mode}. It begins in Prediction mode and switches to Secretary mode upon observing a candidate \( i \) such that the prediction error satisfies \( \left|1 - \frac{\hat{v}_i}{v_i} \right| \ge \theta \), indicating that the predictions is unreliable (Line~\getrefnumber{line:switch} of \Cref{alg:cosp}).
Here, $\theta$ specifies the maximum acceptable prediction error. 
The algorithm seeks to allow as many predictions as possible to be considered reliable, 
while still guaranteeing consistency whenever all errors are within $\theta$.
In our algorithm, if a mode switch occurs, the algorithm transitions to Secretary mode at the time denoted in this paper as \( t_M \). It then follows a variant of Dynkin’s rule, the well-known $1/e$-competitive strategy that 
skips all candidates arriving before a threshold time~$\tau$ and hires the first candidate thereafter 
whose value exceeds all previously observed ones.

%\hbcomment{consider citing Dynkin, or explain Dynkin's rule earlier in the paper.} 
Our algorithm skips all candidates arriving before a time threshold \( \tau \), and selects the first candidate thereafter with a value higher than %who is better than 
all previously observed ones. Special care is taken if \( \hat{i} \) arrives in Secretary mode: it is hired with probability \( \gamma \) if it triggers the mode switch (i.e., arrives at time \( t_M \)), and with probability \( \delta \) if it arrives after the mode switched. These randomized rules ensure robustness, avoiding over-reliance on potentially incorrect predictions while still preserving high performance when predictions are accurate. Moreover, when \(\theta\) is large enough, the algorithm is less likely to switch to secretary mode and more likely to hire \(\hat{i}\), relying on prediction and hiring the highest predicted value candidate \(\hat{i}\).

\begin{comment}
    
\Cam{Again this is repeating something that has been said many times before.}
A key difference between our algorithm and that of Fujii and Yoshida~\cite{Fujii2023} is the introduction of probabilistic hiring rules governed by the parameters \( \gamma \) and \( \delta \), which determine how the algorithm handles the top predicted candidate \( \hat{i} \) in Secretary mode. This modification is essential for robustness %\Cam{Its not essential -- Fujii and Yoshida achieve robustness without it.}
In chosen order setting: while near-consistency requires selecting \( \hat{i} \) when predictions are accurate, there are adversarial scenarios where deterministically hiring \( \hat{i} \) leads to poor performance. 

As an illustrative example, consider the case where a single candidate deviates from its prediction. Suppose this candidate is actually the best overall and also has the highest predicted value among those seen so far. In this case, it is optimal to accept the candidate upon seeing its actual value. However, in a similar scenario, the highest prediction might correspond to the second-best candidate overall. If one selects it based on its high value, it is possible that the true best candidate appears later and is missed due to the earlier selection. This motivates us to pick it with a defined probability to balance between these two scenarios.

\end{comment}
Later, we show that with carefully chosen parameters, our algorithm for COSP setting achieves a competitive ratio of at least
\begin{equation*}
 \max\left\{ \lowerCOSP, \frac{1 - \epsilon}{1 + \epsilon} \right\}
\end{equation*}
across all instances with maximum multiplicative prediction error \( \epsilon \).

\begin{algorithm}[t]
\caption{Prediction-Aware Secretary Algorithm for COSP }
%\Helia{will fix this Saturday}
\label{alg:cosp}
\begin{algorithmic}[1]
\Require 
    Predicted values \( \hat{v}_i \) for all \( i \in [n] \); 
    threshold parameter \( \theta \in [0,1] \), \( \tau \in [0,1] \); fixed arrival time \(\beta\) of the top predicted candidate; probabilities \( \gamma, \delta \in [0,1]: \)  %\hbcomment{I would drop the word ``threshold.'' I think it's confusing because $\theta$ is error threshold, $\tau$ is time threshold, and $\beta$ is not a threshold.} \hbcomment{Algorithm is nice otherwise!}%\Cam{I don't olike the term 'fallback' here.}
%\Comment{$\gamma$ = prob. of hiring \( \hat{i} \) if it triggers Secretary mode; $\delta $= prob. of hiring if \( \hat{i} \) arrives after \hspace*{2em} mode switch} \Cam{I think if going to explain two of the parameters in a comment then shounjld explain all.}
\State \( \hat{i} \gets \arg\max_{i \in [n]} \hat{v}_i \) \Comment{Top predicted candidate}
\State Assign \( t_{\hat{i}} \gets \beta \); assign \( t_i \sim \text{Unif}(0,1) \) for all \( i \neq \hat{i} \)
\State \( \text{mode} \gets \texttt{Prediction} \)
\For{each candidate \( i \in [n] \) in increasing order of arrival time \( t_i \)}
    \If{\( \left|1 - \hat{v}_i / v_i \right| > \theta \) \textbf{and} \( \text{mode} = \texttt{Prediction} \)} \label{line:switch}
        \State \( \text{mode} \gets \texttt{Secretary} \); \( t_M \gets t_i \) \Comment{First large deviation triggers switch}
    \EndIf
    \If{\( \text{mode} = \texttt{Prediction} \) \textbf{and} \( i = \hat{i} \)}
        \State \Return hire candidate \( i \)
    \ElsIf{\( \text{mode} = \texttt{Secretary} \) \textbf{and} \( t_i > \tau \) \textbf{and} \( v_i > \max\{v_j : t_j < t_i\} \)} %\Helia{Need to double check with others} %\Cam{Here is another place I think we are assuming unique values. So need to mention that assumption in prelims.}
        \If{\( i = \hat{i} \) \textbf{and} \( t_i = t_M \)} \label{line:gamma}
            \State Hire \( i \) with probability \( \gamma \)
        \ElsIf{\( i = \hat{i} \)} \label{line:delta}
            \State Hire \( i \) with probability \( \delta \)
        \Else 
            \State Hire candidate \( i \)\label{line:normal}
        \EndIf
    \EndIf
\EndFor
\end{algorithmic}
\end{algorithm}
%\Cam{Code looks good to me.}

\subsection{Competitive Ratio Guarantee}

%\Helia{will do it on Saturday.} 
%\Cam{Need an introduce sentence to the subsection saying what you will do in it.}
In this subsection, we formally analyze the performance of our algorithm by establishing a lower bound on its competitive ratio. Specifically, we divide the analysis into multiple cases and prove that the algorithm achieves a worst-case competitive ratio of at least \(\max\left\{0.262, \frac{1 - \epsilon}{1 + \epsilon} \right\}\) for the Chosen Order Secretary Problem (\COSP). %\hbcomment{there is a space printed after COSP. I tried putting ``xspace'' in the macro, but it messes up elsewhere.}

% State Theorem 1 and parameters.
\begin{theorem}\label{thm:bgd-competitive}
    ~\Cref{alg:cosp} with parameters $\theta = \lowerCOSPtheta$, $\tau = \lowerCOSPtau$, $\beta = \lowerCOSPbeta$, $\gamma = \lowerCOSPgamma$, and $\delta = \lowerCOSPdelta$  is at least $\max\left\{\lowerCOSP, \frac{1 - \epsilon}{1 + \epsilon}\right\}$-competitive for the Chosen Order Secretary Problem (\COSP). % with prediction, assuming a pre-defined arrival time $\beta$ for the candidate with the highest predicted value.
\end{theorem}

\begin{proof}
Following the case-based analysis of Fujii and Yoshida~\cite{Fujii2023}, we divide the proof into seven cases, depending on the structure of the high-deviation set $M$, which contains all candidates whose predictions deviate from their actual values by more than the threshold $\theta$ used to switch from Prediction to Secretary mode in Line~\getrefnumber{line:switch} of \Cref{alg:cosp} %\Cam{I switched this ref to use cleverref. You should start using that package too. I added it to the main file with some default settings.}
. Formally, we define this set as:
\begin{equation*}
  M = \left\{ i \in [n] : \left|1 - \frac{\hat{v}_i}{v_i} \right| > \theta \right\}.  
\end{equation*}
Our cases will consider whether the candidate with the highest predicted value \( \hat{i} \) and/or the candidate with the highest true value \( i^* \) belong to $M$. Also, we define $i_M$ as the first candidate's index to be seen that falls in $M$ and we denote its arrival time as $t_M$.
We also let \({ k = \left| \{ i \in [n] \mid v_i > v_{\hat{i}} \} \right| }\) denote the number of candidates with true value greater than \( \hat{i} \), and let \( m_2 = \left| \{ i \in M \mid v_i < v_{\hat{i}} \} \right| \) denote the number of candidates in set $M$ with true value less than \( \hat{i} \).

Throughout the analysis, we discuss the regime where the fixed arrival time $\beta$ of $\hat{i}$ 
satisfies $\beta > \tau$, as this aligns with parameters obtained in our 
optimization over parameters $\tau,\beta,\gamma,\delta$ (~\cref{thm:bgd-competitive}). This focus simplifies 
the presentation, but both conditions $\beta > \tau$ and $\beta < \tau$ are fully analyzed in the 
next section to establish guarantees across all parameter choices.
 %Throughout the analysis, we assume that the fixed arrival time \( \beta \) of \( \hat{i} \) satisfies \( \beta > \tau \). \hbcomment{Say why it is ok to make this assumption.} This assumption simplifies the analysis and aligns with the parameter settings we eventually choose.
%\Cam{All of this is good to mention but feels in a random order. Should take about things relevant to $M$ (like $t_M$, $k$, etc.) all in one place. Then talk about the assumption that $\beta > \tau$ and the numerical optimization.}
We denote the worst-case competitive ratio achieved in each case \( j \in \{0, 1, \ldots, 6\} \) by \( C_j(\cdot) \) respect to any input set by the adversary, where \( C_j(\cdot) \) is a function of parameters set by the adversary (i.e., \(m\), \(k\), \(v_i\), and \(\hat{v}_i\)).
%\Cam{You actually use these as functions that take other parameters (like $m$ and $k$ and input arguments. Should explain that here.} 
The case ordering follows the structure introduced in~\cite{Fujii2023}.

%\Cam{Need to remind the reader that these are the params of \Cref{alg:cosp}. The won't just recognize a list of Greek letters.}
We will formally verify that our algorithm achieves a competitive ratio of at least \(\max\left\{ \lowerCOSP, \frac{1 - \epsilon}{1 + \epsilon} \right\}\) across all cases. For clarity of presentation, we defer the technical lemmas used in the analysis to the appendix, while the main text focuses on the case-by-case argument structure.

To organize the argument, we present each of the seven cases as a separate part, analyzing the conditions under which the algorithm hires a candidate and deriving a lower bound on the competitive ratio in each scenario.

%\hbcomment{Mohammadreza: Please proofread the cases. I fixed the minor issues I saw, but I think it's good to read carefully. In some places, there is no space after the end-of-sentence period.}

  \paragraph{\textbf{Case 0: \(\boldsymbol{|M| = 0}\).}} \label{par:case0}
In this case, all predictions fall within the acceptable error threshold 
(i.e., $\left|1 - \frac{\hat{v}_i}{v_i} \right| < \theta$), so the high-deviation set $M$ is empty, 
and the algorithm remains in prediction mode throughout. 

%--------
\begin{comment}
This implies the maximum multiplicative prediction error $\epsilon$ is less than or equal to the error threshold \( \epsilon \leq \theta \), all predicted values obey % \Helia{is formula 1 always true or only when predictions are good?}
the bound of~\eqref{eq:pred-bound}.%\Dan{always true}

We know that the algorithm selects \( \hat{i} = \arg\max_i \hat{v}_i \), by using \eqref{eq:pred-bound} we have:
\begin{equation}
\label{eq:case0}
v_{\hat{i}} \geq \frac{\hat{v}_{\hat{i}}}{1 + \epsilon} \geq \frac{\hat{v}_{i^*}}{1 + \epsilon} ,
\end{equation}

using Equation~\eqref{eq:pred-bound} again we have:
\begin{equation}
\label{eq:case01}
v_{\hat{i}} \geq \frac{\hat{v}_{i^*}}{1 + \epsilon} \geq \frac{(1 - \epsilon)v_{i^*}}{1 + \epsilon},
\end{equation}

%\Cam{Logic of this equation doesn't look correct. Should first compare $v_{\hat i}$ to $\hat v_{\hat i}$ using \eqref{eq:pred-bound}. Then use that $\hat v_{\hat i} \ge \hat v_{i^*}$. Then compare that to $v_{i^*}$ again using \eqref{eq:pred-bound}.}

which implies a competitive ratio of at least \( \frac{1 - \epsilon}{1 + \epsilon}\).
Since \( \epsilon \leq \theta = 0.58 \), we have
\begin{equation}
\frac{1 - \epsilon}{1 + \epsilon} > \lowerCOSP = 0.262.
\end{equation}
\end{comment}
%----------
This implies that the maximum multiplicative prediction error satisfies 
$\epsilon \leq \theta$, and by the Multiplicative Error Guarantee 
(Equation~\eqref{eq:pred-bound}), all predicted values remain within 
the $(1 \pm \epsilon)$ bound. 

Since the algorithm selects $\hat{i} = \arg\max_i \hat{v}_i$, it follows that
\begin{equation}
\label{eq:case0}
v_{\hat{i}} \;\geq\; \frac{\hat{v}_{\hat{i}}}{1 + \epsilon} \;\geq\; \frac{\hat{v}_{i^*}}{1 + \epsilon} ,
\end{equation}
and moreover,
\begin{equation}
\label{eq:case01}
v_{\hat{i}} \;\geq\; \frac{\hat{v}_{i^*}}{1 + \epsilon} 
\;\geq\; \frac{(1 - \epsilon)v_{i^*}}{1 + \epsilon}.
\end{equation}

Together, these inequalities imply a competitive ratio of at least 
$\tfrac{1 - \epsilon}{1 + \epsilon}$. Since $\epsilon \leq \theta = 0.58$, we have
\begin{equation}
\frac{1 - \epsilon}{1 + \epsilon} \geq \frac{1 - 0.58}{1 + 0.58}  >  \lowerCOSP .
\end{equation}

%----------
Thus, in this case the algorithm achieves a competitive ratio of 
\(\tfrac{1 - \epsilon}{1 + \epsilon}\), which is at least 
\(\max\!\left\{ \lowerCOSP, \tfrac{1 - \epsilon}{1 + \epsilon} \right\}\), as claimed, since \(\epsilon \leq \theta\).
%\Cam{This is a bit confusing. It is $\frac{1-\epsilon}{1+\epsilon} = \max\left\{ \lowerCOSP, \frac{1 - \epsilon}{1 + \epsilon} \right\}$ since $\epsilon \le \theta$. Let's say that explictly. Else the max sort of comes out of nowhere.}
% \Cam{I didn't follow this statement. Don't we need to say something like in this case the max is achieved by $1-\epsilon/1+\epsilon$? Since $(1-.58)/(1+.58) > .26$. And thus the claimed competitive ratio is achieved?}
%----------------------------------------------------------------------------
%----------------------------------------------------------------------------
%----------------------------------------------------------------------------

%\Cam{I find it weird to use Roman  numerals for the case numbers but then we label competitive ratios with arabic numbers $C_0,\ldots,C_6$ I would just use arabic numbers for both.}
  \paragraph{\textbf{Case 1: \(\boldsymbol{\hat{i} = i^* \in M}\).}} \label{par:case1}
Here, the top predicted candidate \( \hat{i} \) is also the true optimum \( i^* \), but it lies in the high-deviation set \( M \). The algorithm switches to secretary mode upon encountering the first candidate in \( M \), which may or may not be \( i^* \) itself.

We analyze two subcases, depending on whether \( i^* \) triggers the mode switch or not:

%\Cam{Formating is weird here. Use non-indents and smallskips or something like that to start each case.}
\noindent \emph{Case (1.1) \( \beta = t_M \)}: That is, \(i^*\) is the first candidate in \( M \), and it arrives at time \( \beta \).
The probability of this subcase occurring is equal to the probability that all other candidates in \(M\) arrive after \(\beta\), which is \((1 - \beta)^{m - 1}\). Also, the probability of successfully selecting \(i^*\) is \(\gamma\) since \(i^*\) is the first observed error. Thus, the contribution of subcase \emph{Case(1.1)} to the competitive ratio of this case is \(\gamma (1-\beta)^{m-1}\): the event that \(i^*\) is the first observed member of $M$ with probability \((1-\beta)^{m-1}\), and conditioned on this event the algorithm selects it with probability \(\gamma\) (~\cref{alg:cosp}, Line~11).

%In this case, no previous candidate could have triggered Secretary mode or been selected, so no condition on earlier arrivals is required \Cam{Don't know what 'no condition on earlier arrivals is required' means.}. The algorithm hires \( i^* \) with probability \( \gamma \), and all other candidate in $M$ must arrive after \( \beta \), which occurs with probability \( (1 - \beta)^{m - 1} \). \Cam{I would first say what the probability of the case is. Then say the probability of hiring in that case. Then multiply together to give the probability of being in that case AND hiring. Also need to ref back to the algorithm. E.g. ref to the line where you hire with probability $\gamma$ and explain why it is that line which is triggered. As written you are sort of assuming the reader has the algorithm in their head, which they don't.}

\noindent\emph{Case(1.2) \( \beta \neq t_M \)}: Here, some other candidate in \( M \) arrived before $\hat{i}$, so we are already in Secretary mode by the time \( i^* \) arrives at time \( \beta > \tau \).
The probability of this subcase happening is equal to \(1-(1 - \beta)^{m - 1}\) which is at least one candidate in $M$ arrives before $\hat{i}$.
Now, assuming we are in this subcase, in order to reach $\hat{i}$ (to accept it), the algorithm must not accept anyone before $\hat{i}$. It is sufficient (but not necessary) that the best candidate in the interval \([0, \beta)\) came before $\tau$, as then no one would be hired before it. The probability of this scenario is \( \tau / \beta \). Furthermore, the Algorithm accepts $\hat{i}$ with probability $\delta$ according to Line~\getrefnumber{line:delta} of \Cref{alg:cosp}.
It is worth mentioning that the arrival of the best candidate in $[0, \beta)$ before $\tau$ is independent of whether a candidate in $M$ arrives before $\hat{i}$, since candidates in $M$ arriving before $\beta$ are also uniformly distributed in $[0, \beta)$. 

%To ensure that \( i^* \) is considered \Cam{Not clear what  `considered' means. Say explciitly that it means taht no-one is hired before time $\beta$.}, the best candidate in the interval \([0, \beta)\) must have arrived before time \( \tau \), or else the algorithm would have terminated early \Cam{This doesn't sound right. It may not have terminated early if it was still in prediction mode for example. I think you want to say that it is sufficient (but not neccesary) that the second best candidate came before $\tau$, as then no-one would be hired before $i^*$.}. The probability of this is \( \tau / \beta \). Additionally, at least one candidate in $M$ must have arrives before \( \beta \), which happens with probability \( 1 - (1 - \beta)^{m - 1} \). The algorithm then hires \( i^* \) with probability \( \delta \).\Cam{Again need to ref back to algorithm and say way. And again restructured as in case (i) to first give probability of case, then of hiring in the case, then total probability.}

%\Cam{Above you need to argue that the event of the second best candidate in $[0,\beta]$ arriving before $\tau$ is independent of the event taht at least one mistake happens before $\beta$.}

Combining both subcases, the competitive ratio in Case 1 for any  $m$ is lower-bounded by:
\begin{equation}
\label{eq:case1}
    C_1(m,\cdot) \geq \frac{\tau}{\beta} \cdot \delta \cdot \left(1 - (1 - \beta)^{m - 1} \right)
    + \gamma \cdot (1 - \beta)^{m - 1}.
\end{equation}

  \paragraph{\textbf{Case 2: \(\boldsymbol{\hat{i} = i^* \notin M}\).}} \label{par:case2}
Under this condition, the top predicted candidate \( \hat{i} \) is also the true optimum \( i^* \), and its prediction is accurate enough that it does not belong to the high-deviation set \( M \). The algorithm may or may not have entered secretary mode before observing \( i^* \), depending on whether another candidate in $M$ has been seen earlier.

To obtain a pessimistic lower bound in this case, we assume that \( i^* \in M \), even though by the conditions of case 2 we know \( i^* \notin M \). This assumption decreases the chance of hiring \( i^* \). %, since it removes the possibility of certain selection in prediction mode, and therefore yields a valid lower bound. \Cam{This is way to informal. Need to formally say why the bound holds. You do this below, so maybe just remove the informal argument above and give the formal argument below.}
 If the algorithm is already in secretary mode when it sees \( i^* \), then whether or not \( i^* \in M \) has no effect---the selection behavior and probabilities are the same (see Line~\getrefnumber{line:delta} of \Cref{alg:cosp}). %\Cam{Why? Need to ref to algorithm and explain.} 
However, if the algorithm is still in prediction mode when \( i^* \) arrives, then: (1) If \( i^* \notin M \), the algorithm selects it deterministically. (2) If \( i^* \in M \), the algorithm does not select it immediately and instead switches to secretary mode, where it will only be selected with probability \( \gamma \leq 1 \). Thus, the pessimistic assumption \( i^* \in M \) reduces the chance of hiring \( i^* \), and we conclude the result is always better than \textit{Case~1}, presented above, with one extra candidate in \( M \):
\begin{equation}
    C_2(m,\cdot) \geq C_1(m + 1,\cdot).
\end{equation}

%----------------------------------------------------------------------------
%----------------------------------------------------------------------------
%----------------------------------------------------------------------------
  \paragraph{\textbf{Case 3: \(\boldsymbol{\hat{i} \neq i^*,~ \hat{i} \in M,~ i^* \in M}\).}} \label{par:case3}
In this situation, both the top predicted candidate \( \hat{i} \) and the true optimum \( i^* \) belong to the high-deviation set \( M \). %, so the algorithm switches to secretary mode upon encountering the first such candidate in $M$. \Cam{This 'so the algorithm' claim doesn't make sense. That is always true, regardless of the case we are in.}   
We lower bound this case by
%We compare this case to the harder 
Case 4, where $\hat{i} \neq i^*$, $\hat{i} \in M$, and $\hat{i} \notin M$. The only difference between the two cases is that here $i^* \in M$. We further analyze this case by considering two subcases depending on whether $t^*$ is greater than or less than $\beta$. %\hbcomment{I would not call the other cases easier or harder. We know what we mean by that, but it may not be clear from a reader's perspective.}

\noindent\emph{Case(3.1):} If \( t^* > \beta \), the algorithm is already in secretary mode when \( i^* \) arrives because $\hat i$ arrives at time $\beta$ is in set $M$, so whether or not \( i^* \in M \) does not affect its eligibility for hiring. The algorithm selects \( i^* \) if it is the best seen so far and no better candidate has already been selected in Line~\getrefnumber{line:normal} of \Cref{alg:cosp}. %\Cam{Again the word 'considers' is vague here. Be more concrete. Does it hire?}

\noindent\emph{Case(3.2):} If $t^* < \beta$, then for $i^*$ to be selected, the algorithm must already be in secretary mode at or before time $t^*$. When $i^* \in M$, it can itself trigger this mode switch. However, if $i^* \notin M$, it can only be selected if another candidate in $M$ appears earlier. Hence, in all scenarios, $i^* \notin M$ is selected with probability no greater than when $i^* \in M$.

%\Helia{what does this sentence say?}. \Cam{I get what you are saying but this needs to be make more clear/formal.}

Therefore, Case~3 is no harder than Case~4 with one fewer member in $M$ (because of removing $i^*$ from $M$ and thus decreasing its size by $1$), and we conclude that for a fixed $m$, the worst-case competitive ratios of the two cases satisfy the following relation:
\begin{equation*}
    C_3(m,\cdot) \geq C_4(m - 1,\cdot).
\end{equation*}
%\Cam{Should be more explicit in these types of bounds why you are going from $m$ to $m-1$, or $m$ to $m+1$, etc. I.e. here say explicitly that you would be removing $i^*$ from $M$ and thus decreasing its size by $1$.}

%\Hedyeh{can make this (III) a bit more formal.}

%----------------------------------------------------------------------------
  \paragraph{\textbf{Case 4: \(\boldsymbol{\hat{i} \neq i^*,~ \hat{i} \in M,~ i^* \notin M}\).}} \label{par:case4}
Under this condition, the top predicted candidate \( \hat{i} \) belongs to the high-deviation set \( M \), while the true optimum \( i^* \notin M \) has a reliable prediction. Therefore, the algorithm never selects a candidate in prediction mode and switches to secretary mode upon encountering the first candidate \( i \in M \). We analyze two main subcases based on the arrival time $t^*$ of $i^*$. The first subcase occurs when $t^* \in (\tau, \beta)$, and the second when $t^* \in (\beta, 1]$. We further divide the second subcase into two parts: (i) the best candidate in $[0, t^*)$ is not $\hat{i}$, and (ii) the best candidate in $[0, t^*)$ is $\hat{i}$. Combining the analyses of all these subcases yields the overall competitive ratio for Case~4.

\noindent\emph{Case(4.1)} When \( t^* \in (\tau, \beta) \):
In this subcase, the algorithm can only consider \( i^* \) if it has already switched to secretary mode by time \( t^* \). This requires that at least one candidate from \( M \) to have arrived before \( t^* \), which occurs with probability \( 1 - (1 - t^*)^{m - 1} \), also, it is sufficient (but not necessary) that the best candidate in \( [0, t^*) \) arrives before \( \tau \) as then no one would be hired before $i^*$, which happens with probability \( \tau / t^* \). %\Cam{Similar to a previous comment: I don't think this is a *neccesary condition*. But it is sufficient. Also again need to say *why* they are independent. Not just assert it.}
Similar to Case~1, these two probabilities are independent since any candidate in $M$ which comes before $t^*$ also uniformly arrives between $[0,t^*)$. We obtain the bound:
\begin{equation}
\label{eq:c41q}
C_{4.1}(m,\cdot) \geq \int_{\tau}^{\beta} \left(1 - (1 - t^*)^{m-1} \right) \cdot \frac{\tau}{t^*} \, dt^*.
\end{equation}
using~\cref{lem:lemma41} to calculate this integral, we get:

\begin{equation}
\label{eq:c41}
C_{4.1}(m,\cdot) \geq \tau \sum_{i=1}^{m-1} \binom{m-1}{i} (-1)^{i+1} \cdot \frac{\beta^i - \tau^i}{i}.
\end{equation}

%Due to the independence of these two events, and  \Cam{Never ref forward to someting in a paper. If you are going to use it it needs to have already been stated. Perfonally, I would probably call this (and similar lemmas) `Facts' and maybe state them in a prelims section without proof, giving the proofs in the appendix.} we obtain the bound:

%\Cam{The notation $C_4^{(1)}$ is undefined. ALso for this case 1 don't get why the competitive ratio is just a function of $m$. It should be a function of $m, m_2$, and $k$ right?}

\noindent\emph{Case(4.2) When \( t^* \in (\beta, 1] \):}
In this subcase, the algorithm has already seen \( \hat{i} \in M \), and we examine whether \( i^* \) can still be observed and selected under different conditions. We further divide this case into two subcases:

\noindent\emph{Case(4.2.1)} The best candidate in \( [0, t^*) \) is not \( \hat{i} \):
This occurs with probability $1 - (1 - t^*)^k$, where $k$ is the number of candidates whose true value exceeds that of $\hat{i}$. For $i^*$ to be selected, it is sufficient that the best candidate in $[0, t^*)$ arrives before $\tau$, which happens with probability $\tau / t^*$. Similar to previous cases, these two events are independent, since knowing that some candidate better than $\hat{i}$ arrives before $t^*$ does not affect the fact that all arrivals before $t^*$ are uniformly distributed over $[0, t^*)$. Thus, we obtain the following equation:
\begin{align}
\label{eq:c42q}
C_{4.2.1}(m.\cdot) 
&\geq \int_{\beta}^{1}\left(1 - (1-t^*)^k \right) \cdot \frac{\tau}{t^*} 
 dt^* .
\end{align}
Using~\cref{lem:lemma42} to calculate the integral, we have: %\Cam{$C_4^{(2)}(m)$ is undefined. Also need to talk about independent here.} \Cam{In general I couldn't follow the argument here. I didn't see where you were using that the best candidate was not $\hat i$. Needs to be explained more. Just because the best candidate is not $\hat i$ doesn't mean that $\hat i$ won't be hired right? Since that candidate could arive in $[\beta,t^*]$.}
\begin{align}
\label{eq:c42}
C_{4.2.1}(m) 
&\geq  \tau \sum_{i=1}^{k} \binom{k}{i} (-1)^{i+1} \cdot \frac{1 - \beta^i}{i}.
\end{align}

\noindent\emph{Case(4.2.2)} The best candidate in \( [0, t^*) \) is \( \hat{i} \):
In this scenario, which occurs with probability $(1 - t^*)^k$, the algorithm must reach time $t^*$ without hiring $\hat{i}$ in order for $i^*$ to be selected. This leads to two possible outcomes:

\noindent\emph{Case(4.2.2.1)} %Here, soome other candidate
A candidate in \( M \) other than $\hat{i}$ triggers the switch to secretary mode before \( \beta \): This occurs with probability $1 - (1 - \beta)^{m_2}$, where $m_2$ is the number of candidates in $M$ whose values are less than that of $\hat{i}$. We use $m_2$ here because all candidates in $M$ with values greater than $\hat{i}$ must arrive after $t^*$, given that $\hat{i}$ is assumed to be the best candidate in $[0, t^*)$. %\Cam{Need to explain why $m_2$ arises here.}
It is further required that the best candidate in $[0, \beta)$ arrives before $\tau$, which happens with probability $\tau / \beta$. In this case, we do not select any candidate before seeing $\hat{i}$, and when $\hat{i}$ arrives, it must also not be hired, which occurs with probability $1 - \delta$ (Line~\getrefnumber{line:delta} of \Cref{alg:cosp}). Since the event of having a candidate arrive before $\beta$ is independent of the fact that arrivals in $[0, \beta)$ are uniformly distributed, any best candidate in $[0, \beta)$ arrives uniformly at random. Considering statement \emph{Case(4.2.2)} and multiplying these independent probabilities, we have:

\begin{equation}
\label{eq:c43q}
\begin{aligned}
C_{4.2.2.1}(m) &\geq \int_{\beta}^{1} (1 - t^*)^k 
\left(1 - (1 - \beta)^{m_2} \right)(1 - \delta)\frac{\tau}{\beta} 
dt^*.
\end{aligned}
\end{equation}

\noindent\emph{Case(4.2.2.2)} Candidate \( \hat{i} \) itself triggers the switch to secretary mode: This event occurs with probability \( (1 - \beta)^{m_2} \), and is not hired with probability \( 1 - \gamma \). Therefore, incorporating the factor $(1 - t^*)^k$ from (4.2.2), we have:
\begin{equation}
\label{eq:c43w}
\begin{aligned}
C_{4.2.2.2}(m) &\geq \int_{\beta}^{1} (1 - t^*)^k 
(1 - \beta)^{m_2}(1 - \gamma)
 dt^* .
\end{aligned}
\end{equation}

Combining Equations \eqref{eq:c43q} and \eqref{eq:c43w} %\Cam{Should number them as equations and then ref back to them.}
and using~\cref{lem:lemma43} to solve the integrals, we have: %\Helia{Should not we say why some terms are summed up and some are multiplied?}
\begin{equation}
\label{eq:c43}
\begin{aligned}
C_{4.2.2}(m) &\geq \int_{\beta}^{1} (1 - t^*)^k \left[
\left(1 - (1 - \beta)^{m_2} \right)(1 - \delta)\frac{\tau}{\beta} + 
(1 - \beta)^{m_2}(1 - \gamma)
\right] dt^* \\
&= \frac{(1 - \beta)^{k+1}}{k+1} \cdot \left[
\left(1 - (1 - \beta)^{m_2} \right)(1 - \delta)\frac{\tau}{\beta} + 
(1 - \beta)^{m_2}(1 - \gamma)
\right].
\end{aligned}
\end{equation}
Combining Equations~\eqref{eq:c41}, \eqref{eq:c42}, and \eqref{eq:c43} from the subcases, we obtain the total lower bound for this case: %\Helia{should not we say why equations are 'summed up'?} \Cam{Yeah you should. Perhaps you should define the subcases first, and write the total competitive ratio as the sum of probabilities of hiring and being in each subcase. Then bound those probabilities then plug back into the intial sum equation to get the result.}
\begin{align}
    \label{eq:case4}
    C_4(m) \geq\;
        &\tau \sum_{i=1}^{m-1} \binom{m-1}{i} (-1)^{i+1} \cdot \frac{\beta^i - \tau^i}{i} \\
        &+ \tau \sum_{i=1}^{k} \binom{k}{i} (-1)^{i+1} \cdot \frac{1 - \beta^i}{i} 
        \notag \\
        &+ \frac{(1 - \beta)^{k+1}}{k+1} \cdot \left[ \left(1 - (1 - \beta)^{m_2} \right)(1 - \delta)\frac{\tau}{\beta} + (1 - \beta)^{m_2}(1 - \gamma) \right]
     \notag
\end{align}

%----------------------------------------------------------------------------
%----------------------------------------------------------------------------
%----------------------------------------------------------------------------
  \paragraph{\textbf{Case 5: \(\boldsymbol{\hat{i} \neq i^*,~ \hat{i} \notin M,~ i^* \in M}\).}} \label{par:case5}
In this case, the top predicted candidate \( \hat{i} \) is accurate and not in the high-deviation set \( M \), while the true optimum \( i^* \in M \). The algorithm switches to secretary mode upon observing the first candidate from \( M \). We analyze the outcome based on the arrival time \( t^* \) of \( i^* \), and the timing of the best candidate before \( t^* \).

%\hbcomment{Helia: Generally, throughout the paper, a lot of our cases start with ``Here''. It would be good to change that.}

%\emph{Case(i)} 
\noindent\emph{Case(5.1)} \( t^* \in [\tau, \beta] \): In this subcase, the competitive ratio depends on whether \( i^* \) is the first candidate in \( M \), i.e., whether it triggers the switch to secretary mode. This subcase is further divided into two subcases:

%\emph{Case(i.a)} 
\noindent\emph{Case(5.1.1)} \( i^* = t_M \): Under this condition, no other candidate in \( M \) arrives before \( t^* \), and \( i^* \) is the first observed candidate in $M$. 
%observed. 
In this case, the algorithm switches to secretary mode at \( t^* \), and immediately selects \( i^* \). This happens with probability \( (1 - t^*)^{m - 1} \), since all other \( m - 1 \) candidates in \( M \) must arrive after \( t^* \).

\noindent\emph{Case(5.1.2)} \( i^* \neq t_M \): In this scenario, some other candidate from \( M \) must arrive before \( t^* \), causing the algorithm to switch modes before seeing \( i^* \). In order for \( i^* \) to be selected, it must arrive after the switch, and the best candidate in \( [0, t^*) \) must have arrived before \( \tau \), which occurs with probability \( \tau / t^* \). Also, the probability that at least one candidate in $M$ arrives before \( t^* \) is \( 1 - (1 - t^*)^{m - 1} \). Hence, the contribution of this case is
%this case contributes a term of
\(
\left(1 - (1 - t^*)^{m - 1} \right) \cdot \frac{\tau}{t^*}.
\)
Combining both subcases \noindent\emph{Case(5.1.1)} and \noindent\emph{Case(5.1.2)} and using~\cref{lem:lemma51} to solve the integral, we have:
\begin{align}
\label{eq:C51}
C_{5.1}(m,\cdot) \geq \int_{\tau}^{\beta} \left[
\left(1 - (1 - t^*)^{m-1} \right) \cdot \frac{\tau}{t^*} + (1 - t^*)^{m-1}
\right] dt^* \notag \\
=\; \tau \sum_{i=1}^{m-1} \binom{m-1}{i} (-1)^{i+1} \cdot \frac{\beta^i - \tau^i}{i} 
+ \frac{(1 - \tau)^m - (1 - \beta)^m}{m}.
\end{align}

%\emph{Case(ii)} 
\noindent\emph{Case(5.2)} \( t^* \in [\beta, 1] \): Here the optimal candidate arrives after the predicted top candidate. We distinguish two subcases based on whether the best candidate in \( [0, t^*) \) is \( \hat{i} \) or not.

%\emph{Case(ii.a)} 
\noindent\emph{Case(5.2.1)} The best candidate in \( [0, t^*) \) is not \( \hat{i} \): In this subcase the algorithm may still consider \( i^* \) under certain favorable conditions. First, let us analyze the probability that \( \hat{i} \) is not the best candidate in \( [0, t^*) \). This occurs with probability \( 1 - (1 - t^*)^k \), where \( k \) is the number of candidates with true value exceeding that of \( \hat{i} \). Now, conditioned on this event, there must exist some candidate in $M$ arriving before \( \beta \). Knowing this, the probability that at least one candidate in \( M \) arrives before \( \beta \) is at least \( 1 - (1 - \beta)^{m - 1} \). This is because: 

The presence of a better candidate than $\hat{i}$ before \( t^* \) can only increase the likelihood that a candidate from \( M \) also appears before \( \beta \), or at worst, leaves it unchanged.

Therefore, even under this conditioning, we can safely lower-bound the probability of having some candidate in $M$ before \( \beta \) by \( 1 - (1 - \beta)^{m - 1} \).

Finally, in order to observe \( i^* \), the best candidate in \( [0, t^*) \) must arrive before \( \tau \), which occurs independently with probability \( \tau / t^* \). Thus, the combined lower bound is  (which can be rewritten using~\cref{lem:lemma52}):
\begin{align}
\label{eq:C52}
C_{5.2.1}(m,\cdot) &\geq \int_{\beta}^{1} \left(1 - (1 - t^*)^k \right) \cdot \frac{\tau}{t^*} \cdot \left(1 - (1 - \beta)^{m - 1} \right) dt^* \notag \\
&= \tau \sum_{i=1}^{k} \binom{k}{i} (-1)^{i+1} \cdot \frac{1 - \beta^i}{i} \cdot \left(1 - (1 - \beta)^{m - 1} \right).
\end{align}

%\Hedyeh{Explain the correlation and how to handle it.}

\noindent\emph{Case(5.2.2):} The best candidate in \( [0, t^*) \) is \( \hat{i} \): In this subcase, the algorithm must avoid selecting \( \hat{i} \) during secretary mode in order to eventually observe and select \( i^* \). This happens with the following conditions are met: 
%\hbcomment{I would take out of the bullet points format, but maybe number them (1), (2), (3), (4), etc.}
\begin{comment}
\begin{itemize}
    \item The probability that \( \hat{i} \) is the best candidate in \( [0, t^*) \) is \( (1 - t^*)^k \), where \( k \) is the number of candidates with true value greater than \( v_{\hat{i}} \).
    \item For the algorithm to enter secretary mode before seeing \( \hat{i} \), at least one candidate in $M$ from the subset of \( M \) with value below \( v_{\hat{i}} \) must arrive before \( \beta \). This happens with probability \( 1 - (1 - \beta)^{m_2} \), where \( m_2 \) is the number of such candidates.
    \item To avoid selecting \( \hat{i} \), the algorithm must skip it, which happens with probability \( 1 - \delta \).
    \item Finally, the best candidate before $\beta$ must have arrived before \( \tau \), so that the algorithm does not terminate before observing \( \hat{i} \). This occurs with probability \( \tau / \beta \).
\end{itemize}
\end{comment}

\begin{enumerate}
    \item The probability that \( \hat{i} \) is the best candidate in \( [0, t^*) \) is \( (1 - t^*)^k \), 
    where \( k \) is the number of candidates with true value greater than \( v_{\hat{i}} \).
    \item For the algorithm to enter secretary mode before seeing \( \hat{i} \), at least one candidate in \( M \) 
    from the subset of \( M \) with value below \( v_{\hat{i}} \) must arrive before \( \beta \). 
    This happens with probability \( 1 - (1 - \beta)^{m_2} \), where \( m_2 \) is the number of such candidates.
    \item To avoid selecting \( \hat{i} \), the algorithm must skip it, which happens with probability \( 1 - \delta \).
    \item Finally, the best candidate before \( \beta \) must have arrived before \( \tau \), 
    so that the algorithm does not terminate before observing \( \hat{i} \). 
    This occurs with probability \( \tau / \beta \).
\end{enumerate}

Combining all these probabilities using the contribution from this setting is lower-bounded by:
\begin{align}
\label{eq:C53}
C_{5.2.2}(m) &\geq \int_{\beta}^{1} (1 - t^*)^k \cdot \left[ \left(1 - (1 - \beta)^{m_2} \right)(1 - \delta)\frac{\tau}{\beta} \right] dt^* \notag \\
&= \frac{(1 - \beta)^{k+1}}{k+1} \cdot \left[ \left(1 - (1 - \beta)^{m_2} \right)(1 - \delta)\frac{\tau}{\beta} \right].
\end{align}
Summing the contributions from cases \emph{Case(5.1)} and \emph{Case(5.2)} in Equations~\eqref{eq:C51}, \eqref{eq:C52}, and~\eqref{eq:C53}, the total competitive ratio for \textit{Case~5} is lower-bounded by:
\begin{align}
\label{eq:case5}
C_5(m,\cdot) \geq\;
    &\tau \sum_{i=1}^{m-1} \binom{m-1}{i} (-1)^{i+1} \cdot \frac{\beta^i - \tau^i}{i}
    + \frac{(1 - \tau)^m - (1 - \beta)^m}{m} \notag \\
    &+ \tau \sum_{i=1}^{k} \binom{k}{i} (-1)^{i+1} \cdot \frac{1 - \beta^i}{i} \cdot \left(1 - (1 - \beta)^{m - 1} \right) \notag \\
    &+ \frac{(1 - \beta)^{k + 1}}{k + 1} \cdot \left[ \left(1 - (1 - \beta)^{m_2} \right)(1 - \delta)\frac{\tau}{\beta} \right].
\end{align}

%----------------------------------------------------------------------------
%----------------------------------------------------------------------------
%----------------------------------------------------------------------------

  \paragraph{\textbf{Case 6: \(\boldsymbol{\hat{i} \neq i^*,~ \hat{i} \notin M,~ i^* \notin M}\).}} \label{par:case6}
In this case, neither the top-predicted candidate $\hat{i}$ nor the true optimal candidate $i^*$ belongs to the high-deviation set $M$. Consequently, the algorithm may either remain entirely in prediction mode or switch to secretary mode, depending on whether any element of $M$ arrives before $\hat{i}$. The overall analysis thus combines contributions from \emph{Case~4} (when $t^* \in [\tau, \beta]$), \emph{Case~5} (when $t^* \in [\beta, 1]$), together with an additional term accounting for successful selection in prediction mode.

\noindent\emph{Case(6.1):} If the algorithm stays in prediction mode and hires \( \hat{i} \), this occurs when all members of \( M \) arrive after \( \beta \), which happens with probability \( (1 - \beta)^m \). Since \( \hat{i} \notin M \), its predicted value is within a multiplicative factor \( \epsilon \) of its true value. Thus, by Equation~\eqref{eq:pred-bound}, the competitive ratio is at least \( \frac{1 - \epsilon}{1 + \epsilon} \). Therefore, the contribution to the competitive ratio from this event is:
\begin{equation}
\label{eq:C60}
C_{6.1}(m,\cdot) = (1 - \beta)^m \cdot \frac{1 - \epsilon}{1 + \epsilon}.
\end{equation}

\noindent\emph{Case(6.2):} If the algorithm switches to secretary mode before encountering \( \hat{i} \), then the selection of \( i^* \) depends on its arrival time \( t^* \).

%\emph{(ii.a)} 
\noindent\emph{Case(6.2.1):} \( t^* \in [\tau, \beta] \): Here, the algorithm is already in secretary mode, and the same logic as Equation~\eqref{eq:c41} from \emph{case 4} applies:
\begin{equation}
\label{eq:C62}
C_{6.2.1}(m,\cdot) \geq \tau \sum_{i=1}^{m} \binom{m}{i} (-1)^{i+1} \cdot \frac{\beta^i - \tau^i}{i}.
\end{equation}

%\emph{Case(ii.b)} 
\noindent\emph{Case(6.2.2)} \( t^* \in [\beta, 1] \): Here, we follow the same analysis as in Equations~\eqref{eq:C52} and~\eqref{eq:C53} from \emph{Case~5}, where $m$ denotes the number of candidates in $M$, and $m_2$ denotes the number of candidates in $M$ whose true values are less than $v_{\hat{i}}$, respectively. We obtain the following two contributions:

\begin{equation}
\label{eq:C63}
C_{6.2.2}^{(1)}(m,\cdot) \geq \tau \sum_{i=1}^{k} \binom{k}{i} (-1)^{i+1} \cdot \frac{1 - \beta^i}{i} \cdot \left(1 - (1 - \beta)^{m} \right),
\end{equation}

\begin{equation}
\label{eq:C64}
C_{6.2.2}^{(2)}(m,\cdot) \geq \frac{(1 - \beta)^{k+1}}{k+1} \cdot \left[ \left(1 - (1 - \beta)^{m_2} \right)(1 - \delta)\frac{\tau}{\beta} \right].
\end{equation}

Summing the contributions from \emph{Case(6.1)}, \emph{Case(6.2.1)}, and \emph{Case(6.2.2)}, we obtain the overall lower bound for Case~6:

\begin{align}
\label{eq:case6}
C_6(m,\cdot) \geq\;
    & (1 - \beta)^m \cdot \frac{1 - \epsilon}{1 + \epsilon} \notag \\
    &+ \tau \sum_{i=1}^{m} \binom{m}{i} (-1)^{i+1} \cdot \frac{\beta^i - \tau^i}{i} \notag \\
    &+ \tau \sum_{i=1}^{k} \binom{k}{i} (-1)^{i+1} \cdot \frac{1 - \beta^i}{i} \cdot \left(1 - (1 - \beta)^{m} \right) \notag \\
    &+ \frac{(1 - \beta)^{k+1}}{k+1} \cdot \left[ \left(1 - (1 - \beta)^{m_2} \right)(1 - \delta)\frac{\tau}{\beta} \right].
\end{align}

\paragraph{Completing the Proof via Case Enumeration.}
Recall that for each of the structural cases analyzed above, we derived an explicit
competitive-ratio expression $C_j(m,k,m_2)$, parameterized by the integers
$m$, $k$, and $m_2$ that describe the adversarial instance.  
The algorithm's overall guarantee on a given instance is
\[
C(m,k,m_2) \;=\; \min_{j\in\{0,\ldots,6\}} C_j(m,k,m_2),
\]
and our goal is to show that
\[
C(m,k,m_2)\;\ge\; \max\left\{\lowerCOSP,\; \frac{1-\epsilon}{1+\epsilon}\right\}
\qquad\text{for all }(m,k,m_2).
\]

A direct minimization of $C(m,k,m_2)$ over all integer triples $(m,k,m_2)$ is not
tractable: the closed-form expressions contain alternating sums and exponentially
small terms such as $(1-\beta)^m$ and $(1-t^\ast)^k$, which makes it difficult to analyze these formulas simultaneously for all integers. 
However, these expressions simplify dramatically once $m$, $k$, or $m_2$ becomes
sufficiently large.  For the remaining small values of $(m,k,m_2)$, the minimum of
the case formulas can be computed exactly using a finite grid, following an approach
similar to that of~\cite{Fujii2023}.

This motivates a structured enumeration strategy: we partition the parameter space into eight ``regimes'' depending on whether each variable is  \emph{small} (at most $20$) or \emph{large} (greater than $20$).   Regime~1 corresponds to all three parameters being small, while Regimes~2--8  cover all remaining combinations in which at least one parameter is large.

For small values ($m,m_2,k\le 20$), we can evaluate each $C_j(m,k,m_2)$
exactly using its closed form.  
For large values, terms such as $(1-\beta)^m$ become negligible (e.g.,
for $m>20$ and $\beta=0.64$, we have $(1-\beta)^m<10^{-4}$), and can be conservatively replaced by zero.  
These replacements only \emph{decrease} the computed competitive ratio, so they yield valid symbolic lower bounds.  
With this approach, the verification reduces to a finite enumeration over the eight regimes, as summarized in Algorithm~\ref{alg:lower-bound-proof}.

\begin{algorithm}
\caption{Lower Bound Proof via Case Analysis}
\label{alg:lower-bound-proof}
\begin{algorithmic}[1]
\Require Formulas $C_1,C_4,C_5,C_6$ parameterized by $(m,k,m_2)$
\Require Thresholds $T_m = 20$, $T_k = 20$
\Require Target lower bound $B \gets \max\{\lowerCOSP,\; \tfrac{1-\epsilon}{1+\epsilon}\}$

\State \textbf{Regime 1: All parameters small}
\For{$m = 0$ to $T_m$}
    \For{$k = 0$ to $T_k$}
        \For{$m_2 = \max(0,m-k)$ to $m$}
            \State $v \gets \min(C_1(m,k,m_2),\, C_4(m,k,m_2),\, C_5(m,k,m_2),\, C_6(m,k,m_2))$
            \State \textbf{Assert:} $v \ge B$
        \EndFor
    \EndFor
\EndFor

\State \textbf{Regimes 2--8: At least one parameter large}
\ForAll{combinations of large/small status for $(m,k,m_2)$ except Regime 1}
    \State Substitute the symbolic lower bounds appropriate for that regime
    \State $L \gets \min(C_1, C_4, C_5, C_6)$ under this regime
    \State \textbf{Assert:} $L \ge B$
\EndFor

\State \textbf{Conclusion:} All admissible $(m,k,m_2)$ satisfy $C(m,k,m_2)\ge B$
\end{algorithmic}
\end{algorithm}

The enumeration confirms that for all instances the competitive ratio is at least
%\[
$\max\{\lowerCOSP,\allowbreak\; \tfrac{1-\epsilon}{1+\epsilon}\}$,
%\]
when using the parameters
%\[
$\theta=\lowerCOSPtheta,\;%\qquad 
\tau=\lowerCOSPtau,\;%\qquad 
\beta=\lowerCOSPbeta,\;%\qquad
\gamma=\lowerCOSPgamma,\;%\qquad
\delta=\lowerCOSPdelta.$
%\]

%\Cam{This section needs more explaination and needs to be more formal. E.g. you could write someting like $$C \ge \min_{j,m,k,m_2} C_j(m,k,m_2),$$ and then break up this min into the ranges of values that we manually check and those covered by the large value cases.}
% we perform an exhaustive case-based lower bound check across all combinations of $(m, k, m_2)$. We distinguish between \textbf{small} and \textbf{large} values of these parameters. 

\subparagraph{Symbolic Lower Bounds for Large Parameters.}
We now list the symbolic simplifications for large $m$, $k$, and $m_2$ used in Regimes 2–8 (regimes to 8 cases based on $m$, $k$, and $m_2$ being small or large, i.e. regime 1 is all of them being small) of~\cref{alg:lower-bound-proof}.
%\Cam{the reader has no idea at this point what Algorithm 2 is an why it is relevant here.} 
We have to calculate other regimes cause formulas are dependent to $m,m_2$ and $k$ and we should bound them. For example, assuming a large $m>20$ in the Equation~\cref{eq:case1}, we can ignore the right-hand side and left-hand side $(1-\beta)^{m-1}$ is less than 0.0001, so the left-hand side will be bounded by $(\tau/\beta) \times \delta \times 0.9999$ same idea works for Equation~\cref{eq:c41q} and other similar equations. In general, in the following formulas, we rounded down all small numbers to $0$. and All $1-\text{small number}$ to 0.9999. As we can show, all these small numbers are less than $0.0001$.
Similarly, for other large $m>20$ in the exponent, the corresponding terms approach zero, allowing us to simplify the equations as follows:

\emph{For large $m > 20$:}
\begin{align}
C_1(.) &\geq 0.9999 \cdot \frac{\tau}{\beta} \cdot \delta \\
C_{4.1}(.) &\geq 0.9999 \cdot \tau \cdot \ln\left(\frac{\beta}{\tau}\right) \\
C_{5.1}(.) &\geq 0.9999 \cdot \tau \cdot \ln\left(\frac{\beta}{\tau}\right) \\
C_{6.1}(.) &\geq 0 \\
 C_{6.2.1}(.) &\geq 0.9999 \cdot \tau \cdot \ln\left(\frac{\beta}{\tau}\right)
\end{align}

The next parameter to discuss is $k$. Similar to $m$, when we raise fixed numbers to the power of $k > 20$, the result will be close to zero.

\emph{For large $k > 20$:}
\begin{align}
C_{4.2.1}(m,.) &\geq 0.9999 \cdot \tau \cdot \ln\left(\frac{1}{\beta}\right) \\
C_{4.2.2}(m,.) &\geq 0 \\
C_{5.2.1}(m,.) &\geq 0.9999 \cdot \tau \cdot \ln\left(\frac{1}{\beta}\right) \cdot \left(1 - (1 - \beta)^{m-1}\right) \\
C_{5.2.2}(m,.) &\geq 0 \\
C_{6.2.2}^{(1)}(m,.) &\geq 0.9999 \cdot \tau \cdot \ln\left(\frac{1}{\beta}\right) \cdot \left(1 - (1 - \beta)^{m}\right) \\
C_{6.2.2}^{(2)}(m,.) &\geq 0
\end{align}
The last parameter is $m_2$, which is bounded by $m$. However, when $m$ is large, this parameter can also become large. By rewriting the equation for large $m_2$, in the same way as in the previous two parts, we obtain:

\emph{For large $m_2$:}
\begin{align}
C_{4.2.2}(.) &\geq
\frac{(1 - \beta)^{k+1}}{k+1} \cdot \left[ 0.9999 \cdot (1 - \delta) \cdot \frac{\tau}{\beta} \right] \\
C_{5.2.2}(.) &\geq
\frac{(1 - \beta)^{k+1}}{k+1} \cdot \left[ 0.9999 \cdot (1 - \delta) \cdot \frac{\tau}{\beta} \right] \\
C_{6.2.2}^{(2)}(.) &\geq
\frac{(1 - \beta)^{k+1}}{k+1} \cdot \left[ 0.9999 \cdot (1 - \delta) \cdot \frac{\tau}{\beta} \right]
\end{align}
We should also consider the case where both $m$ and $k$ are large, and determine which formula to examine. 
It is worth mentioning that we do not need to calculate the case where both $m_2$ and $k$ are large separately, 
since $m_2$ being large implies that $m$ is large, which simplifies all equations sufficiently to compute.

\emph{For large $m$ and $k$:}
\begin{align}
C_{5.2.1}(.) &\geq 0.9999^2 \cdot \tau \cdot \ln\left(\frac{1}{\beta}\right) \\
C_{6.2.1}(.) &\geq 0.9999^2 \cdot \tau \cdot \ln\left(\frac{1}{\beta}\right)
\end{align}

Using all the above observations in \cref{alg:lower-bound-proof} establishes that the chosen parameters guarantee the competitive ratio, and thus the theorem holds.

\end{proof}

\section{Random Order Secretary Problem (ROSP) with Prediction}
\label{sec:rand}
%\Helia{all $\beta$s should be replaced with $\hat{t}?$}\\
%\Helia{similar to section 3, should not we say why some formulas are summed op or multiplied?}\\
%\Helia{Should not we refer to some lemmas?}\\
%\hbcomment{In this section, there were some places that we called the positioning of $\hat{i}$ in COSP ``adversarial''. This is not correct. I changed all that I saw. Please double-check there are no such occurrences.}

We now analyze the same algorithm from ~\cref{sec:tbgd-method} under a randomized arrival setting, where each candidate—including the top predicted candidate \(\hat{i}\)—is assigned an independent arrival time \(t_i \sim \text{Unif}[0,1]\), as in the classical secretary model. This replaces the %adversarial 
optimally chosen positioning of \(\hat{i}\) with a uniform distribution, aligning with the standard assumption in prior work such as Fujii and Yoshida~\cite{Fujii2023}. 

The analysis is divided into seven cases, indexed by \( j \in \{0, 1, \ldots, 6\} \), with the corresponding competitive ratio denoted by the function \( C'_j(\cdot) \) of parameters set by the adversary (i.e., \(m\), \(k\), \(v_i\), and \(\hat{v}_i\)).
 While the algorithm in this section remains unchanged from \cref{alg:cosp}---beginning in prediction mode and switching to secretary mode upon observing a candidate with large deviation from its prediction---the analysis is significantly impacted. Specifically, any term involving \(\beta\), the pre-determined arrival time of \(\hat{i}\), must now be averaged over \([0,1]\) and replaced with $\hat{t}$, introducing an additional layer of integration into the competitive ratio bounds for each case. Similar to ~\cref{sec:tbgd-method}, we numerically optimize the parameters of \Cref{alg:cosp}---\( \theta, \tau, \gamma, \delta \)---and formally verify that, under random-order input, our algorithm achieves a competitive ratio of at least $
\max\left\{\lowerROSP,\; \frac{1 - \epsilon}{1 + \epsilon}\right\}$ across all cases.

%\hbcomment{I changed the structure of this theorem to that of theorem 1.}
\begin{theorem}\label{thm:bgd-rand}
%\hbedit{
~\Cref{alg:cosp} with parameters \(\theta = \lowerROSPtheta\), \(\tau = \lowerROSPtau\), \(\gamma = \lowerROSPgamma\), \(\delta = \lowerROSPdelta\) and $\beta$ uniformly drawn from $[0,1]$ is at least \(
\max\left\{\lowerROSP,\; \frac{1 - \epsilon}{1 + \epsilon}\right\}\text{-competitive}
\)-competitive for the Random Order Secretary Problem (ROSP). %}
%If we set \(\theta = \lowerROSPtheta\), \(\tau = \lowerROSPtau\), \(\gamma = \lowerROSPgamma\), and \(\delta = \lowerROSPdelta\), then~\Cref{alg:cosp} from ~\cref{sec:tbgd-method} is 
%\(
%\max\left\{\lowerROSP,\; \frac{1 - \epsilon}{1 + \epsilon}\right\}\text{-competitive}
%\)
%for the classical secretary problem with prediction.
\end{theorem}

\begin{proof}
%As in the adversarial-position case analyzed in ~\cref{sec:tbgd-method},
Similar to ~\cref{sec:tbgd-method}, we divide the analysis into seven cases depending on whether the top predicted candidate \(\hat{i}\) and the true optimal candidate \(i^*\) are in the set \( M = \left\{ i \in [n] : \left|1 - \frac{\hat{v_i}}{v_i}\right| > \theta \right\} \).
The structure of the analysis and the algorithmic behavior in each case remains the same. The key difference is that the arrival time $\hat{t}$ of \(\hat{i}\) is now drawn uniformly at random from \([0,1]\), introducing an additional averaging step into all expressions that depend on $\hat{t}$. Let \(i^* = \arg\max_{i \in [n]} v(i)\) denote the optimal candidate. Similar to ~\cref{sec:tbgd-method}, we define \(k = \left| \left\{ i \in [n] : v_i > v_{\hat{i}} \right\} \right|\) as the number of candidates whose true value exceeds that of \(\hat{i}\), and \(m_2 = \left| \left\{ i \in M : v_i < v_{\hat{i}} \right\} \right|\) as the number of candidates with unreliable prediction with true value less than \(\hat{i}\). We now analyze each case in the same order as in ~\cref{sec:tbgd-method}, integrating over the uniform distribution of time $\hat{t}$ where appropriate. As in ~\cref{sec:tbgd-method}, the supporting technical lemmas are deferred to the appendix, and we present only the main case-by-case arguments in this section.

\paragraph{\textbf{Case 0: \(\boldsymbol{|M| = 0}\).}}  

This case mirrors \emph{Case 0} from the analysis in \cref{sec:tbgd-method}. Since no candidate exceeds the prediction error threshold~\(\theta\), the algorithm remains in prediction mode throughout and directly selects the top predicted candidate~\(\hat{i}\). As shown in Equations~\eqref{eq:case0} and~\eqref{eq:case01}, the competitive ratio in this case is at least \(\tfrac{1 - \epsilon}{1 + \epsilon}\). Moreover, because \(|M| = 0\) implies \(\epsilon \leq \theta = \lowerROSPtheta\), and since the final competitive guarantee is defined as the maximum between this bound \(\tfrac{1 - \epsilon}{1 + \epsilon}\) and the constant bound derived from the remaining cases (at least~\(\lowerROSP\)), we conclude that this case meets the required performance threshold.

%----------------------------------------------------------------------------
%----------------------------------------------------------------------------
%----------------------------------------------------------------------------
In the remainder of the proof, we consider the case where \(M \neq \emptyset\), so the algorithm transitions from prediction to secretary mode upon encountering the first candidate in \(M\). As stated above, here all arrival times---including that of \(\hat{i}\)---are drawn independently from the uniform distribution over \([0,1]\). To handle this, for each case, we partition the analysis based on whether \(\hat{t} < \tau\) or \(\hat{t} \geq \tau\). For range \(\hat{t} \in [\tau, 1]\), we reuse the case-specific lower bounds from ~\cref{sec:tbgd-method} and integrate them over \(\hat{t}\), replacing each occurrence of \(\beta\) with the random variable \(\hat{t}\).
 For range \(\hat{t}\in [0, \tau]\), we derive new expressions reflecting the algorithm's behavior in the early-arrival setting and likewise integrate them. Summing these two contributions for each case completes the competitive ratio analysis required to prove~\cref{thm:bgd-rand}.
%\Helia{I asked Hedyeh about below comment:}
%\hbcomment{Throughout all files, it would be good to search for all ’ and replace with '}

\paragraph{\textbf{Case 1: \(\boldsymbol{\hat{i} = i^* \in M}\).}} 
Here, the candidate with the highest predicted value is also the true optimal candidate, but it has a large deviation from its prediction and hence belongs to set \(M\). The algorithm will switch to secretary mode upon encountering the first candidate \(i \in M\), and may select \(\hat{i} = i^*\) during that phase. Depending on where $\hat{t}$ falls, we examine two scenarios:

\noindent\emph{(1.1)} When \(\hat{t} < \tau\): 
The probability of selecting \(i^*\) in this regime is zero. Since \(i^* = \hat{i}\) arrives before the threshold \(\tau\), the algorithm will not hire it in prediction mode. Moreover, regardless of whether \(i^*\) is the first member of \(M\) or not, the algorithm will not consider it during secretary mode either, because it appears before $\tau$. %before the switching point.

\noindent\emph{(1.2)} When \(\hat{t} \in [\tau, 1]\):  
As in the pre-determined \(\hat{t} = \beta\) setting (see Equation~\eqref{eq:case1} in ~\cref{sec:tbgd-method}), the algorithm selects \(\hat{i} = i^*\) with probability \(\delta\) if it arrives after the switch (i.e., is observable), and with probability \(\gamma\) if it is the switch-triggering candidate, corresponding to Lines~11 and~13 of \Cref{alg:cosp}, respectively.

Since time \(\hat{t}\) is now uniformly distributed over \([0,1]\), we integrate lower bound Equation~\eqref{eq:case1} %of case 1 
over the interval \([\tau, 1]\):
\begin{equation}
C_1'(m,.) \geq \delta \tau \int_{\tau}^{1} \frac{1 - (1 - \hat{t})^{m - 1}}{\hat{t}} \, d\hat{t}
+ \gamma \int_{\tau}^{1} (1 - \hat{t})^{m - 1} \, d\hat{t}.
\label{eq:ccp11}
\end{equation}

We simplify each term in the right-hand side of the above inequality as follows:
\begin{equation}
\begin{aligned}
\delta\tau\int_{\tau}^{1} \frac{1 - (1 - \hat{t})^{m - 1}}{\hat{t}} \, d\hat{t}
&= \delta\tau\ln\left(\frac{1}{\tau}\right)
- \delta\tau\int_{\tau}^{1} \frac{(1 - \hat{t})^{m - 1}}{\hat{t}} \, d\hat{t} \\
&= \delta\tau\ln\left(\frac{1}{\tau}\right)
- \delta\tau\left( \sum_{i = 1}^{m - 1} \binom{m - 1}{i} (-1)^i \cdot \frac{1 - \tau^i}{i}
+ \ln\left(\frac{1}{\tau}\right) \right) \\
&= - \delta\tau\sum_{i = 1}^{m - 1} \binom{m - 1}{i} (-1)^i \cdot \frac{1 - \tau^i}{i},
\label{eq:ccp12}
\end{aligned}
\end{equation}
And for the other term, we have:
\begin{equation}
\begin{aligned}
\gamma\int_{\tau}^{1} (1 - \hat{t})^{m - 1} \, d\hat{t}
&= \gamma\frac{(1 - \tau)^m}{m}.
\label{eq:ccp13}
\end{aligned}
\end{equation}
Therefore, the overall competitive ratio in this case is lower bounded by sum of Equations~\eqref{eq:ccp12} and~\eqref{eq:ccp13}:
\begin{equation}
C_1'(m,.) \geq
\delta \tau \sum_{i = 1}^{m - 1} \binom{m - 1}{i} (-1)^{i+1} \cdot \frac{1 - \tau^i}{i}
+ \gamma \cdot \frac{(1 - \tau)^m}{m}.
\label{eq:cpone}
\end{equation}

%----------------------------------------------------------------------------
%----------------------------------------------------------------------------
%----------------------------------------------------------------------------
\paragraph{\textbf{Case 2: \(\boldsymbol{\hat{i} = i^* \notin M}\).}} 
In this situation, the candidate with the highest predicted value, \(\hat{i}\), is also the true optimal candidate \(i^*\), and its prediction error is small enough that it does not belong to set \(M\). Nevertheless, the algorithm may enter secretary mode before time \(\hat{t}\), triggered by the arrival of some other candidate \(i \in M\), at time \(t_M < \hat{t}\).

Once in secretary mode, the algorithm can select \(i^*\) only if it arrives after \(\tau\) and is the best candidate observed up to its arrival. To conservatively lower bound the performance, we follow the same pessimistic reduction used in \textit{Case~2} of ~\cref{sec:tbgd-method}: we pessimistically assume that \(i^*\) belongs to \(M\), effectively reducing this case to \textit{Case~1} with one additional member in $M$.

This reduction cannot increase the algorithm's performance and may reduce it, since it potentially triggers an earlier switch to secretary mode, limiting the probability of hiring \(i^*\). Since in prediction mode, the algorithm selects the predicted-best candidate 
deterministically at its scheduled arrival time $\hat{t}$. 
At this point, no candidate from $M$ has yet appeared (see Line~7 of \Cref{alg:cosp}).

%\hbcomment{The last sentence is not very clear.} 
Declaring \(i^* \in M\), however, eliminates this possibility and bounds the hiring probability to at most \(\gamma\) or \(\delta\), depending on whether it triggers the mode switch or not. Therefore, the competitive ratio in this case is at least as large as the lower bound from \textit{Case~1} above with \(m+1\) member in $M$:
\begin{equation}
C_2'(m,.) \geq C_1'(m + 1,.).
\end{equation}

%----------------------------------------------------------------------------
%----------------------------------------------------------------------------
\paragraph{\textbf{Case 3: \(\boldsymbol{\hat{i} \neq i^*,\; \hat{i} \in M,\; i^* \in M}\).}} 

In this scenario, both the top predicted candidate \(\hat{i}\) and the true optimal candidate \(i^*\) have large deviation from their predictions, so both belong to set \(M\). We analyze this setting by a pessimistic reduction to \textit{Case~4} (handled next) where $i_{\beta} \neq i^*,~ i_{\beta} \in M$, and \(i^* \notin M\).In particular, we show that since $i^* \notin M$ delays the algorithm's switch to secretary mode, the scenario in \textit{Case~4} has a smaller competitive ratio. %is strictly more challenging. 
We divide the analysis based on the arrival time \(t^*\) of \(i^*\), following the same structure as in ~\cref{sec:tbgd-method}:

\noindent\emph{(3.1)}~If \(t^* > \hat{t}\): then the algorithm is already in secretary mode, triggered by the arrival of \(\hat{i} \in M\) at time \(\hat{t}\). Whether \(i^*\) belongs to \(M\) or not does not affect its selection in this mode; hence, membership in \(M\) is irrelevant in this scenario.

\noindent\emph{(3.2)}~If \(t^* < \hat{t}\): then for the algorithm to observe and select \(i^*\), it must have already switched to secretary mode before time \(t^*\), since we assumed $t^* \notin M$. This occurs only if another member of \(M\) arrives prior to \(t^*\). Declaring \(i^* \notin M\) (as in \textit{Case~4}) reduces the chance of hiring \(i^*\), since otherwise no member of \(M\) would be required to arrive before \(t^*\) and $i^*$ itself could serve as the switch-triggering candidate. %This makes \textit{Case~4} more challenging. 
This results in a smaller competitive ratio for \textit{Case~4}.

By this pessimistic reduction, we conclude that this case is no more difficult than the \textit{Case~4} below with one fewer candidate in $M$. Thus, the competitive ratio is bounded as below:
\begin{equation}
C_3'(m,.) \geq C_4'(m - 1,.).
\end{equation}

%----------------------------------------------------------------------------

\paragraph{\textbf{Case 4: \(\boldsymbol{\hat{i} \neq i^*,~ \hat{i} \in M,~ i^* \notin M}\).}}  
Here, the top predicted candidate \(\hat{i}\) is in \(M\), while the true optimal candidate \(i^*\) is not. Therefore, the algorithm always switches to secretary mode and never hires during prediction mode. We divide the analysis into two parts depending on where \(\hat{t}\) falls:

\noindent\emph{(4.1)} When \(\hat{t} < \tau\): 
In this subcase, \(\hat{i}\) arrives before the threshold \(\tau\), and the algorithm switches to secretary mode at or before time \(\hat{t}\). The only way to select \(i^*\) is if it arrives after \(\tau\), and the best candidate in \([0, t^*)\) arrives before \(\tau\), ensuring this best candidate was not selected during prediction mode. We analyze two further situations now:

\noindent\emph{(4.1.1)} \(\hat{i}\) is the best in \([0, t^*)\):
For this scenario to occur, every candidate \(i\) with \(v_i > v_{\hat{i}}\) must lie outside the interval \([0, t^*)\), which happens with probability \((1 - t^*)^k\), where \(k\) denotes the number of candidates whose true value exceeds that of \(\hat{i}\). The algorithm skips \(\hat{i}\) and reaches \(t^*\), now in secretary mode. Since $t^*$ is drawn uniformly at random, we must integrate this probability 
over the range $t^* \in [\tau,1]$ to capture all possible arrival times of $i^*$ when it can be selected in this scenario and calculate the contribution of this part to the competitive ratio of this case. Therefore:
\begin{equation}
C'_{4.1.1}{}(m,.) \geq \int_{\tau}^{1} (1 - t^*)^k \,dt^*
\label{eq:cp411}
\end{equation}
\noindent\emph{(4.1.2)} Another candidate, rather than $\hat{i}$, is the best in \([0, t^*)\):
This event occurs with probability \(1 - (1 - t^*)^k\), corresponding to the existence of at least one candidate with true value greater than that of \(\hat{i}\) arriving before \(t^*\). And for \(i^*\) to be observed and selected, the best candidate prior to \(t^*\) must arrive before \(\tau\), which happens with probability \(\tfrac{\tau}{t^*}\). The terms multiply because the two events are independent, and we integrate 
over $t^* \in [\tau,1]$ for the same reason as above. Thus:
\begin{equation}
C'_{4.1.2}{}(m,.) \geq \int_{\tau}^{1}\frac{\tau}{t^*}(1 - (1 - t^*)^k)dt^*
\label{eq:cp4100}
\end{equation}
\noindent Summing Equations~\eqref{eq:cp411} and~\eqref{eq:cp4100} and integrating over $[0,\tau]$, 
corresponding to the arrival time of $\hat{t}$, yields the competitive ratio for this case:
\begin{equation}
C'_{4.1}(m,.) \geq \int_0^{\tau}\left(\int_{\tau}^{1}\frac{\tau}{t^*}\left(1 - (1 - t^*)^k\right)dt^* + \int_{\tau}^{1}(1 - t^*)^k \,dt^*\right)d\hat{t}
\label{eq:cp41first}
\end{equation}
using~\cref{lem:lemmaX41} the above equation simplifies to:
\begin{equation}
C'_{4.1}(m,.)
\geq \tau\left[-\tau\sum_{i=1}^{k} \binom{k}{i} (-1)^i \frac{1 - \tau^i}{i} + \frac{(1 - \tau)^{k+1}}{k+1}\right]
\label{eq:cp41}
\end{equation}

\noindent\emph{(4.2)} When \(\hat{t} \geq \tau\):
This scenario follows the same logic as \textit{Case 4} in ~\cref{sec:tbgd-method}. 
We compute the expectation over \(\hat{t} \in [\tau, 1]\), using the COSP expression from 
Equation~\eqref{eq:case4}, to account for the random arrival time of \(\hat{i}\):

\begin{equation*}
    C'_{4.2}(m,.) = \int_{\tau}^{1} C_4(m,\cdot) \, d\hat{t}
\end{equation*}
%Employing Lemma~\cref{lem:lemmaX42} $C'_{4.2}(m)$ simplifies to below equation:
Applying~\cref{lem:lemmaX42} simplifies $C'_{4.2}(m)$:
\begin{equation}
\begin{aligned}
C'_{4.2}(m,.) =
& \tau \sum_{i=1}^{m-1} \binom{m-1}{i} (-1)^{i+1} \left( \frac{1 - \tau^{i+1}}{i(i+1)} - \frac{\tau^i (1 - \tau)}{i} \right) \\
& + \tau \sum_{i=1}^{k} \binom{k}{i} (-1)^{i+1} \left( \frac{1 - \tau}{i} - \frac{1 - \tau^{i+1}}{i(i+1)} \right) \\
& + \frac{(1 - \delta)\tau}{k+1} \left[ 
\sum_{i=1}^{k+1} \binom{k+1}{i} (-1)^i \frac{1 - \tau^i}{i}
- \sum_{i=1}^{k+1+m_2} \binom{k+1+m_2}{i} (-1)^i \frac{1 - \tau^i}{i}
\right] \\
& + \frac{(1 - \gamma)}{k+1} \cdot \frac{(1 - \tau)^{k+2 + m_2}}{k+2 + m_2}
\end{aligned}
\label{eq:cp42}
\end{equation}

\noindent Summing both Equations~\eqref{eq:cp41} and~\eqref{eq:cp42} we have the total competitive ratio of \textit{Case 4}:

\begin{equation}
\begin{aligned}
C'_4(m,.) \geq\;
& \tau \left[-\tau \sum_{i=1}^{k} \binom{k}{i} (-1)^i \frac{1 - \tau^i}{i} + \frac{(1 - \tau)^{k+1}}{k+1} \right] \\
& + \tau \sum_{i=1}^{m-1} \binom{m-1}{i} (-1)^{i+1} \left( \frac{1 - \tau^{i+1}}{i(i+1)} - \frac{\tau^i (1 - \tau)}{i} \right) \\
& + \tau \sum_{i=1}^{k} \binom{k}{i} (-1)^{i+1} \left( \frac{1 - \tau}{i} - \frac{1 - \tau^{i+1}}{i(i+1)} \right) \\
& + \frac{(1 - \delta)\tau}{k+1} \left[ 
\sum_{i=1}^{k+1} \binom{k+1}{i} (-1)^i \frac{1 - \tau^i}{i}
- \sum_{i=1}^{k+1+m_2} \binom{k+1+m_2}{i} (-1)^i \frac{1 - \tau^i}{i}
\right] \\
& + \frac{(1 - \gamma)}{k+1} \cdot \frac{(1 - \tau)^{k+2 + m_2}}{k+2 + m_2}
\end{aligned}
\label{eq:cp4}
\end{equation}

%----------------------------------------------------------------------------
%----------------------------------------------------------------------------
%----------------------------------------------------------------------------

\paragraph{\textbf{Case 5: \(\boldsymbol{\hat{i} \neq i^*,~ \hat{i} \notin M,~ i^* \in M}\).}} 
Here, the top predicted candidate \(\hat{i}\) is not in \(M\), %\(\hat{i} \notin M\), 
while the optimal candidate \(i^*\) is in \(M\).
%\(i^* \in M\). %The algorithm must avoid selecting \(\hat{i}\) in prediction mode. 
Similar to \textit{Case 4} above, we divide the analysis into two ranges for time \(\hat{t}\):

\noindent\emph{(5.1)} When \(\hat{t} < \tau\): For the decision maker to see $t^*$, 
the algorithm must not have selected $\hat{i}$ in prediction mode. Thus, it is necessary that the algorithm has already switched to secretary mode by time \(\hat{t}\); 
otherwise, it would have committed to hiring \(\hat{i}\) upon its arrival.
Therefore, regardless of whether \(\hat{i}\) is the best candidate in \([0, t^*)\), the algorithm must have encountered at least one candidate \(i \in M\) before \(\hat{t}\) to trigger the switch to secretary mode. The probability of this general condition occurring is:

\begin{equation}
    C'_{5.1}(m,.) \;=\; 1 - (1 - \hat{t})^{m_2}.
    \label{eq:cp511}
\end{equation}

\noindent where, as defined earlier, \(m_2\) is the number of candidates in \(M\) with true values less than \(v_{\hat{i}}\). Thus, the overall contribution of \textit{Case~5.1} is obtained by combining Equation~\eqref{eq:cp511} with the two bounds in the subcases below:

%(i,a)
\noindent\emph{(5.1.1)} The best candidate in \([0, t^*)\) is \(\hat{i}\): This condition occurs when all candidates with true values higher than \(\hat{i}\) arrive after \(t^*\). This happens with probability $(1-t^*)^k$, where $k$ is defined as the number of candidates with true value more than $v_{\hat{i}}$. Note that this occurs independently of the general condition of being in secretary mode mentioned in \textit{(5.1)}. Therefore, multiplying by Equation~\eqref{eq:cp511}, this scenario occurs with probability:

\begin{equation}
    C'_{5.1.1}(m,.) = \int_{\tau}^{1}(1-t^*)^k(1 - (1 - \hat{t})^{m_2}) dt^*.
    \label{eq:cp5113}
\end{equation}

%(i.b) 
\noindent\emph{(5.1.2)} The best candidate in \([0, t^*)\) is not \(\hat{i}\): In this scenario, some other candidate \(i\) must have the highest value in \([0, t^*)\) and must also arrive before \(\tau\) to avoid being selected during secretary mode. These events occur with probabilities \(1 - (1 - t^*)^k\) and \(\tfrac{\tau}{t^*}\), respectively, and are independent.
Considering we are in secretary mode, as stated in \textit{(5.1)}, scenario \textit{(5.1.2)} happens with probability:

\begin{equation}
C'_{5.1.2}(m,.) = \int_{\tau}^{1}\frac{\tau}{t^*}(1 - (1 - t^*)^k)(1 - (1 - \hat{t})^{m_2}) dt^*.
    \label{eq:cp5112}
\end{equation}

As before, the integral accounts for averaging over all feasible arrival times \(t^* \in [\tau,1]\).
Summing Equations~\eqref{eq:cp5112} and~\eqref{eq:cp5113} and averaging over the range \([0,\tau]\), where \(\hat{t}\) falls, we obtain the total lower bound for \textit{(5.1)}: %\Helia{Mohammadreza: according to which Lemmas in appendix?}
\begin{equation}
\begin{aligned}
C'_{5.1}(m,.) = \int_0^{\tau} \left( \int_{\tau}^{1}
\left(\frac{\tau}{t^*} \left(1 - (1 - t^*)^k\right) \left(1 - (1 - \beta)^m\right)
+ (1 - t^*)^k \left(1 - (1 - \beta)^{m_2}\right)\right) \, dt^* \right) d\hat{t} \\
= - \tau \left( \sum_{i=1}^{k} \binom{k}{i} (-1)^i \frac{1 - \tau^i}{i} \right)
\left( \tau - \frac{1 - (1 - \tau)^{m+1}}{m+1} \right) \\
\quad + \frac{(1 - \tau)^{k+1}}{k + 1} \left( \tau - \frac{1 - (1 - \tau)^{m_2+1}}{m_2+1} \right)
\end{aligned}
\label{eq:cp51}
\end{equation}
%(ii)
\noindent\emph{(5.2)} When \(\hat{t} \geq \tau\): %\Helia{11, 12, 14}:
  The analysis for this part directly parallels that of \textit{Case~5} in ~\cref{sec:tbgd-method}. As outlined at the beginning of this proof, we only need to compute the expectation over \(\hat{t} \in [\tau, 1]\) of the expression \(C_5(m)\) from Equation~\eqref{eq:case5}, corresponding to the COSP setting in ~\cref{sec:tbgd-method}:
\begin{comment}
\begin{align} 
C_5(m) \geq\;
    &\tau \sum_{i=1}^{m-1} \binom{m-1}{i} (-1)^{i+1} \cdot \frac{\hat{t}^i - \tau^i}{i}
    + \frac{(1 - \tau)^m - (1 - \hat{t})^m}{m} \notag \\
    &+ \tau \sum_{i=1}^{k} \binom{k}{i} (-1)^{i+1} \cdot \frac{1 - \hat{t}^i}{i}\cdot \left(1 - (1-\hat{t})^{m-1} \right) \notag \\
    &+ \frac{(1 - \hat{t})^{k+1}}{k+1} \cdot \left[ \left(1 - (1 - \hat{t})^{m_2} \right)(1 - \delta)\frac{\tau}{\hat{t}}  \right].
\end{align}
\end{comment}
Taking the expectation over \(\hat{t} \in [\tau, 1]\) and solve it using~\cref{lem:lemmaX52}, we get:
\begin{equation}
\begin{aligned}
C'_{5.2}(m,.) &= \int_{\tau}^{1} C_5(m) \, d\hat{t} \\
&\geq \tau \sum_{i=1}^{m-1} \binom{m\!-\!1}{i} (-1)^{i\!+\!1}
\left( \frac{1\!-\!\tau^{i\!+\!1}}{i(i\!+\!1)} - \frac{\tau^i (1\!-\!\tau)}{i} \right) \\
&\quad + \frac{(1\!-\!\tau)^m (1\!-\!\tau)}{m} - \frac{(1\!-\!\tau)^{m\!+\!1}}{m(m\!+\!1)} \\
&\quad + \tau \sum_{i=1}^{k} \binom{k}{i} (-1)^{i\!+\!1} \frac{1}{i}
\left[
(1\!-\!\tau) - \frac{1\!-\!\tau^{i\!+\!1}}{i\!+\!1} - \frac{(1\!-\!\tau)^m}{m}
+ \frac{i!\,(m\!-\!1)!}{(i\!+\!m)!} - \int_0^{\tau} \hat{t}^i (1\!-\!\hat{t})^{m\!-\!1} d\hat{t}
\right] \\
&\quad + \frac{(1\!-\!\delta)\tau}{k\!+\!1} \left[
\sum_{i=1}^{k\!+\!1} \binom{k\!+\!1}{i} (-1)^i \frac{1\!-\!\tau^i}{i}
- \sum_{i=1}^{k\!+\!1\!+\!m_2} \binom{k\!+\!1\!+\!m_2}{i} (-1)^i \frac{1\!-\!\tau^i}{i}
\right]
\end{aligned}
\label{eq:cp52}
\end{equation}
\noindent {Summing Equations~\cref{eq:cp51} and~\cref{eq:cp52} gives the competitive ratio for \textit{Case~5}}:
\begin{equation}
\begin{aligned}
C'_5{}(m,.) \geq\;
& -\tau \left( \sum_{i=1}^{k} \binom{k}{i} (-1)^i \frac{1 - \tau^i}{i} \right)
\left( \tau - \frac{1 - (1 - \tau)^{m+1}}{m+1} \right) \\
& + \frac{(1 - \tau)^{k+1}}{k + 1} \left( \tau - \frac{1 - (1 - \tau)^{m_2+1}}{m_2+1} \right) \\
& + \tau \sum_{i=1}^{m-1} \binom{m-1}{i} (-1)^{i+1}
\left( \frac{1 - \tau^{i+1}}{i(i+1)} - \frac{\tau^i (1 - \tau)}{i} \right) + \frac{(1 - \tau)^{m+1}}{m+1} \\
& + \tau \sum_{i=1}^{k} \binom{k}{i} (-1)^{i+1} \frac{1}{i}
\left[
(1 - \tau) - \frac{1 - \tau^{i+1}}{i+1} - \frac{(1 - \tau)^m}{m}
+ \sum_{j=0}^{m-1} \binom{m-1}{j} (-1)^j \frac{1 - \tau^{i + j + 1}}{i + j + 1}
\right] \\
& + \frac{(1 - \delta)\tau}{k+1} \left[
\sum_{i=1}^{k+1} \binom{k+1}{i} (-1)^i \frac{1 - \tau^i}{i}
- \sum_{i=1}^{k+1+m_2} \binom{k+1+m_2}{i} (-1)^i \frac{1 - \tau^i}{i}
\right]
\end{aligned}
\label{eq:cp5}
\end{equation}
%----------------------------------------------------------------------------
%----------------------------------------------------------------------------
%----------------------------------------------------------------------------

\paragraph{\textbf{Case 6: \(\boldsymbol{\hat{i} \neq i^*,~ \hat{i} \notin M,~ i^* \notin M}\).}} 
In this scenario, neither the top predicted candidate \(\hat{i}\) nor the optimal candidate \(i^*\) is in the mistake set \(M\).
%Here both the top predicted candidate \(\hat{i}\) and the optimal candidate \(i^*\) are not in the mistake set \(M\). 
The algorithm may succeed %either
by remaining in prediction mode and hiring \(\hat{i}\), 
or by switching to secretary mode and subsequently hiring \(i^*\) (or another high-value candidate). 
We analyze these two possibilities separately.

\noindent\emph{(6.1)} %\hbedit{The algorithm remains in prediction mode:} \hbcomment{Another replacement could be
$\hat{i}$ arrives in prediction mode: %The prediction mode remains: 
If there is no candidate \(i \in M\) arriving before time \(\hat{t}\), i.e., the algorithm never switches mode, then it hires \(\hat{i}\) in prediction mode. This occurs with probability \((1 - \beta)^m\). Since both \(\hat{i}\) and \(i^*\) are not in \(M\), the competitive ratio in this case is at least \((1 - \epsilon)/(1 + \epsilon)\), according to Equation~\eqref{eq:pred-bound}. Therefore, the expected contribution is:
\begin{equation}
C'_{6.1}(m,.) = \int_0^1 (1 - \hat{t})^m \cdot \frac{1 - \epsilon}{1 + \epsilon} \, d\hat{t}
= \frac{1 - \epsilon}{1 + \epsilon} \cdot \frac{1}{m + 1}.
\label{eq:cp-prediction}
\end{equation}
Although the algorithm hires \(\hat{i}\) rather than the true optimum \(i^*\), this outcome still counts as a success, because hiring \(\hat{i}\) guarantees a competitive ratio of at least \((1-\epsilon)/(1+\epsilon)\). Note that the probability that $\hat{i}$ arrives before all mistakes is %averaging over the possible arrival times contributes the additional factor
 \(1/(m+1)\).

\noindent\emph{(6.2)} %\hbedit{The algorithm switches to secretary mode:} 
$\hat{i}$ arrives in secretary mode:
%The mode is secretary: 
With probability \(1 - (1 - \hat{t})^m\), the algorithm switches mode and proceeds as in \textit{Case~5} above. We divide this part into two further possible situations:

\noindent\emph{(6.2.1)} Assuming we are already in secretary mode, as described in \textit{(6.2)} above, the best candidate among those arriving before \(t^*\) must fall within the sampling window \([0,\tau]\) for \(t^*\) to be observed and eligible for selection.  
Since arrivals are uniform, the probability of this event is \(\tfrac{\tau}{t^*}\).
Averaging over all possible arrival times of \(i^*\) yields
$\int_{\tau}^1 \frac{\tau}{t^*} \, dt^*$.

Note that the general requirement for mode witch mentioned in \textit{(6.2)} and the conditional success probability given this switch (from \textit{(6.2.1)}) are not independent: the probability in \textit{(6.2.1)} 
is conditional on the algorithm having already switched to secretary mode in \textit{(6.2)}. 
Therefore, the overall competitive ratio of part \textit{(6.2.1)} is obtained by multiplying these two probabilities.

\begin{equation}
C'_{6.2.1}(m,.) =\int_0^{\tau} \left( 1 - (1 - \hat{t})^m \right) \cdot \tau \ln\left( \frac{1}{\tau} \right) \, d\hat{t}
= \tau \ln\left( \frac{1}{\tau} \right) \cdot \int_0^{\tau} \left( 1 - (1 - \hat{t})^m \right) \, d\hat{t}.
\label{eq:cp62}
\end{equation}

\noindent\emph{(6.2.2)} When \(\hat{t} \geq \tau\): This part follows the same reasoning as \textit{Case~6} in ~\cref{sec:tbgd-method}. 
In particular, as noted at the start of the proof, it suffices to evaluate the expectation of 
the COSP expression \(C_6(m)\) from Equation~\eqref{eq:case6} over the interval 
\(\hat{t} \in [\tau,1]\). Considering the general requirement in \textit{(6.2)}, we have:
\begin{equation}
    C'_{6.2.2}(m,.) \;=\; \int_{\tau}^{1} \bigl(1 - (1 - \hat{t})^{m}\bigr) \cdot C_{6}(m,\hat{t}) \, d\hat{t}.
    \label{eq:cp633}
\end{equation}
%\Helia{Mohammadreza: write the solution to above equation and refer to appropriate Lemmas from appendix.}

\noindent Summing Equations~\eqref{eq:cp-prediction},~\eqref{eq:cp62}, and~\eqref{eq:cp633}, we have:
\begin{equation}
\begin{aligned}
C'_6{}(m,.) \geq\;
& \frac{1 - \epsilon}{1 + \epsilon} \cdot \frac{1}{m + 1} \\
& + \tau \ln\left( \frac{1}{\tau} \right) \cdot \int_0^{\hat{t}} \left( 1 - (1 - \hat{t})^m \right) \, d\hat{t} \\
& + \int_{\tau}^{1} \left( 1 - (1 - \hat{t})^m \right) \cdot C_6(m, \hat{t},.) \, d\hat{t}.
\end{aligned}
\label{eq:cp6}
\end{equation}

%\Helia{Mohammadreza: add the solution to Equation~\eqref{eq:cp633}}

%----------------------------------------------------------------------------
%----------------------------------------------------------------------------
%----------------------------------------------------------------------------
%\end{proof}

\paragraph{Final Bound.}
To complete the analysis, we evaluate the minimum competitive ratio across all six cases for small parameter values \(m, m_2, k \leq 20\). 
These cases can be checked directly using the explicit formulas derived above with \(\hat{t}\) integrated uniformly. 
For large values of \(m\), \(k\), or \(m_2\), we apply the same bounding strategy as in \cref{alg:lower-bound-proof} for the COSP analysis: exponential terms such as \((1-\beta)^m\) or \((1-\beta)^k\) vanish as the parameters grow, yielding tractable asymptotic bounds. 
Since \(m_2\) is always bounded by \(m\), there is no need to treat the regime with both \(m_2\) and \(k\) large separately.

\medskip\noindent
Combining the numerical evaluation for small parameters with these asymptotic bounds for large parameters, we confirm that the fixed parameter values
\[
\theta = \lowerROSPtheta,\quad \tau = \lowerROSPtau,\quad \gamma = \lowerROSPgamma,\quad \delta = \lowerROSPdelta
\]
guarantee a competitive ratio of at least
\[
\max\left\{ \lowerROSP,\; \tfrac{1 - \epsilon}{1 + \epsilon} \right\}.
\]
Thus, the theorem holds.%\hbcomment{Mohammadreza: This part is unclear and feels unfinished. Are we done here? Don't we need to give more details about the lower bounds? Also, I \textit{think} here we say `asymptotic' and in the other section we said `symbolic'?}

%----------------------------------------------------------------------------
%----------------------------------------------------------------------------
%----------------------------------------------------------------------------

\end{proof}

\section{Equivalence of Deterministic and Randomized Algorithms in the Classic Setting}
\label{sec:improved-bounds}

In this section, we analyze the secretary problem with predictions in the
classical random–arrival model, where each candidate’s arrival time is
drawn independently and uniformly from $[0,1]$ (\ROSP). We show that in
this model, deterministic and randomized algorithms are equivalent in
power. In particular, we formally prove that any randomized algorithm can
be simulated by a deterministic one by extracting randomness from the
instance itself—specifically, from the arrival time of the first
candidate.
%Unlike the chosen arrival order model (COSP), this setting does not assume any prior knowledge of the arrival time of the predicted best candidate $\hat{i}$.
%\hbcomment{I dropped the last sentence. I think it is clear by now.}

At first glance, one might interpret the results of \cite{Fujii2023} as
suggesting that no deterministic algorithm can achieve a competitive ratio
better than \(0.25\). However, this upper bound is proved for the
\emph{random–order} model, where only the permutation of candidates is
random. In contrast, both our work and \cite{Fujii2023} analyze the
\emph{random–arrival} model, where each arrival time is drawn independently
from the uniform distribution over $[0,1]$. The deterministic upper bound
from the random–order setting does not transfer to the random–arrival
setting. In fact, in the random–arrival model, we show that deterministic
and randomized algorithms are equivalent in power: any randomized
algorithm can be simulated by a deterministic one by extracting
randomness from the instance itself (e.g., from the arrival time of the
first candidate).

%We argue that, in the ROSP setting, deterministic and randomized algorithms are equivalent in power. This contradicts the assumption made in~\cite{Fujii2023}, where the two types of algorithms are treated separately and different upper bounds are claimed (e.g., \(0.25\) for deterministic algorithms). We formally prove that for the secretary problem in the classic setting, \emph{any randomized algorithm can be simulated by a deterministic one}, by extracting internal randomness from the instance itself—specifically, from the arrival time of the first candidate.

\begin{theorem} \label{thm:det=rand}
In the classical setting secretary problem with predictions (\ROSP model), any randomized algorithm can be simulated by a deterministic algorithm without any loss in performance.

\end{theorem}

%\hbcomment{Mohammadreza: Priority: Make the proof more rigorous.}
\begin{proof}
Let $n$ candidates arrive at times drawn independently and uniformly from $[0,1]$. 
Denote their arrival times by $\{t_i : i \in [n]\}$, and let $t_{1} < t_{2} < \cdots < t_{n}$
denote their sorted arrival times (the order statistics). 
Let $i_1 \in [n]$ be the index of the candidate with the earliest arrival, so that $t_{i_1} = \min_{i \in [n]} t_i$.
Note that $t_{i_1}$ is not uniformly distributed over $[0,1]$; rather, it follows the distribution of the minimum of $n$ i.i.d.\ uniform random variables:
\[
\Pr[t_{i_1} \leq x] = 1 - (1 - x)^n.
\]

Let $F(x) = \Pr[t_{i_1} \leq x] = 1 - (1-x)^n$ denote the CDF of $t_{i_1}$.
By the probability integral transform, the random variable
\[
U = F(t_{i_1})
\]
is distributed uniformly on $[0,1]$.
Thus, from the observed value of $t_{i_1}$, the algorithm can compute $U$ deterministically,
and then extract the binary expansion of $U = 0.x_1x_2x_3\ldots$ to obtain an infinite sequence of unbiased random bits.

Therefore, the first arrival time $t_{i_1}$ provides a perfect internal source of uniform randomness.
Any randomized algorithm $A$ that would normally rely on external random bits
can instead simulate them using the bits of $U$.
As a result, we can construct a deterministic algorithm $A'$ that:
\begin{itemize}
    \item observes $t_{i_1}$,
    \item computes $U = F(t_{i_1})$,
    \item extracts the binary expansion of $U$ to simulate the random bits used by $A$,
    \item and then behaves identically to $A$ on all future steps.
\end{itemize}
Since $A'$ and $A$ have identical distributions of outcomes on every input, their competitive ratios are the same.
Hence, in this model, randomized algorithms are no more powerful than deterministic ones.

\end{proof}
\begin{corollary}
Any upper bound on the competitive ratio achievable by deterministic algorithms 
in the ROSP setting also applies to randomized algorithms.
\end{corollary}

\vspace{1em}
\noindent

\section{Concluding Remarks}
In this paper, we studied the secretary problem with predictions under both random-order (ROSP) and chosen-order (COSP) arrivals. We presented a randomized algorithm that follows predictions when they are accurate, but switches to Dynkin’s classical rule when large deviation from prediction occurs, achieving both consistency under accurate predictions and robust guarantees in the worst case. Our analysis shows that the algorithm attains competitive ratios of $\max\{0.221,(1-\epsilon)/(1+\epsilon)\}$ for ROSP and $\max\{0.262,(1-\epsilon)/(1+\epsilon)\}$ for COSP, improving the best-known bounds and going beyond previous barriers in these settings. Recall that \(\epsilon\) is the largest multiplicative error in any prediction. We also proved that randomized and deterministic algorithms are equally powerful in the random-order model, thereby resolving a prior separation claim.

Future research could focus on improving the tightness of the bounds for both models, for instance by approaching the $0.33$ upper limit known for ROSP or closing the gap between our $0.262$ lower bound and the best upper bounds for COSP. Another line of work could be to extend these ideas to other variants such as the $k$-secretary problem or matroid and knapsack constraints. Since decision-makers often have more flexibility in real applications, exploring distributional scheduling of predicted candidates or calibrating algorithmic parameters from data may yield stronger guarantees in practice.

\bibliographystyle{alpha}
\bibliography{references}

\newpage
\appendix
\section{Appendix}
\label{appendix}
This appendix contains the detailed technical proofs and supporting lemmas omitted from the main
text for clarity of presentation. The formulas provided here are general mathematical expressions 
that underpin the analysis, rather than literal restatements of the equations in the main text. 
In particular, we present complete derivations for the case analyses in~\cref{sec:tbgd-method} and~\cref{sec:rand}, along with 
symbolic bounds for large parameter regimes, thereby establishing the rigorous foundations for the 
competitive ratio guarantees stated in~\cref{thm:bgd-competitive} and~\cref{thm:bgd-rand}. For convenience, we reuse the notations introduced in the main body of the paper throughout, 
though not every parameter defined earlier is employed here, and their placement within individual 
formulas may differ from where they were first introduced.

%\hbcomment{This is not a priority, but it would be nice to make this section look nicer. Maybe we can do this after this submission and before the arXiv version.}

%\hbcomment{consider changing the subsection titles or delete some. we probably don't need subsections CDF or PDF. Also, consider having a lemma or observation at the end for the pdf formula.} 

\subsection{Density Function of the Earliest Arrival Time}
%\subsection{density function}
We calculate the density function of 
%the first member of 
the earliest arrival time of members in
set $M$.\\
Let $a_1, a_2, \dots, a_m$ be $m$ independent and uniformly distributed random variables over the interval $[0,1]$. Define the random variable:

\[
X = \min(a_1, a_2, \dots, a_m).
\]

\subsection*{Cumulative Distribution Function (CDF)}

The cumulative distribution function (CDF) of $X$ is given by:

\[
F_X(x) = P(X \leq x).
\]

This represents the probability that at least one of the $m$ numbers is at most $x$. The complementary event is that all $m$ numbers are greater than $x$, which has probability:

\[
P(a_1 > x, a_2 > x, \dots, a_m > x) = (1 - x)^m.
\]

Thus, the CDF is:

\[
F_X(x) = 1 - (1 - x)^m, \quad 0 \leq x \leq 1.
\]

\subsection*{Probability Density Function (PDF)}

The probability density function (PDF) is obtained by differentiating the CDF with respect to $x$:

\[
\begin{aligned}
    f_X(x) &= \frac{d}{dx} F_X(x) \\
    &= \frac{d}{dx} \left(1 - (1 - x)^m \right) \\
    &= m(1 - x)^{m-1}, \quad 0 \leq x \leq 1.
\end{aligned}
\]
The probability density function of $X = \min(a_1, a_2, \dots, a_m)$ is:

\begin{equation}  
f_X(x) = m(1 - x)^{m-1}, \quad 0 \leq x \leq 1.
\end{equation}

\subsection{Lemmas}
\begin{lemma}\label{lemma:lm1}
\begin{equation}
\int_{a}^{b} f_{X}(x)\, dx = (1-a)^m - (1-b)^m.
\end{equation}
\end{lemma}
\begin{proof}
%\Helia{for these integrals: $I_{11},I_{13},I_{21}$ and some others}
What we want to calculate is
\begin{equation*}
J_1 = \int_{a}^{b} m\, (1-x)^{m-1}\, dx.
\end{equation*}
This integral resembles the basic power rule for integration:

\begin{equation*}
\int u^{m-1} \, du = \frac{u^{m}}{m}, \quad \text{for } m \neq 0.
\end{equation*}
Using the substitution \( u = 1 - x \), so that \( du = -dx \), we rewrite the integral as:

\begin{equation*}
J_1 = m \int_{1-b}^{1-a} u^{m-1} du.
\end{equation*}
Applying the power rule:

\begin{equation*}
J_1 = m \left. \frac{u^m}{m} \right|_{1-b}^{1-a}.
\end{equation*}
Substituting the limits \( a \) and \( b \) the final result is::
\begin{equation*}
J_1 = (1 - a)^m - (1 - b)^m.
\end{equation*}
%~~~~~~~~~~~~~~~~~~~~~~~~~~~~~~~~~~~~~~~~~~~~~~~~
\end{proof}
\begin{lemma}\label{lemma:lm2}
\begin{equation}
\int_{a}^{b} x f_{X}(x)\, dx = (1-a)^m - (1-b)^m - \frac{m}{m+1}.(1-a)^{m+1} + \frac{m}{m+1}.(1-b)^{m+1} .
\end{equation}
\end{lemma}
\begin{proof}
%\Helia{for integrals $I_{12},I_{22}$ and some others.}

We start with
\begin{equation*}
J_2 = \int_{a}^{b} \frac{x}{\beta}\, m\, (1-x)^{m-1}\, dx.
\end{equation*}
Factor out the constant:
\begin{equation*}
J_2 = m \int_{a}^{b} x\, (1-x)^{m-1}\, dx.
\end{equation*}
Now, substitute:
\begin{equation*}
u = 1-x, \quad \text{so that} \quad x = 1-u.
\end{equation*}
Also,
\begin{equation*}
du = -\,dx \quad \Longrightarrow \quad dx = -\,du.
\end{equation*}
The limits change as:
\begin{equation*}
\text{When } x = a, \; u = 1-a; \quad \text{when } x = b, \; u = 1-b.
\end{equation*}
Substitute into the integral:
\begin{equation*}
J_2 = m\int_{u=1-a}^{1-b} (1-u)\, u^{m-1}\, (-du).
\end{equation*}
Reverse the limits to remove the negative sign:
\begin{equation*}
J_2 = m\int_{u=1-b}^{1-a} (1-u)\, u^{m-1}\, du.
\end{equation*}
Expand the integrand:
\begin{equation*}
(1-u)\, u^{m-1} = u^{m-1} - u^m.
\end{equation*}
Thus,
\begin{equation*}
J_2 = m \left[ \int_{1-b}^{1-a} u^{m-1}\, du - \int_{1-b}^{1-a} u^m\, du \right].
\end{equation*}
Integrate term by term:
\begin{equation*}
\int u^{m-1}\, du = \frac{u^m}{m},
\end{equation*}
\begin{equation*}
\int u^m\, du = \frac{u^{m+1}}{m+1}.
\end{equation*}
Substitute the antiderivatives:
\begin{equation*}
J_2 = m\left[ \left.\frac{u^m}{m}\right|_{u=1-b}^{1-a} - \left.\frac{u^{m+1}}{m+1}\right|_{u=1-b}^{1-a} \right].
\end{equation*}
Simplify:
\begin{equation*}
J_2 = \left[ (1-a)^m - (1-b)^m - \frac{m}{m+1}\left((1-a)^{m+1} - (1-b)^{m+1}\right) \right].
\end{equation*}

\end{proof}

\begin{lemma}\label{lemma:lm3}
\begin{equation}
\int_{a}^{b} \frac{1}{x}\, dx = ln(b/a).
\end{equation}
\end{lemma}
\begin{proof}
We have
\begin{equation*}
J_3 = \int_{a}^{b} \frac{1}{x}\, dx.
\end{equation*}

\begin{equation*}
J_3 = ln(b) - ln(a) = ln(b/a).
\end{equation*}

\end{proof}

\begin{lemma}\label{lemma:lm4}
\begin{equation}
\int_{a}^{b} \frac{(1-x)^m}{x}\, dx =  \sum_{i=1}^{m} \binom{m}{i} (-1)^{i}\frac{(b^i - a^i)}{i} + \ln\left(\frac{b}{a}\right)
\end{equation}
\end{lemma}
\begin{proof}
We begin by expanding the numerator using the binomial theorem:
\[
(1 - x)^m = \sum_{i=0}^{m} \binom{m}{i} (-1)^i x^i.
\]
Thus,
\[
\frac{(1 - x)^m}{x} = \sum_{i=0}^{m} \binom{m}{i} (-1)^i x^{i - 1}.
\]

Now integrate both sides over the interval \([a, b]\):
\[
\int_a^b \frac{(1 - x)^m}{x} \, dx = \sum_{i=0}^{m} \binom{m}{i} (-1)^i \int_a^b x^{i - 1} \, dx.
\]

We now compute the integral for each term. For \(i = 0\), we have
\[
\int_a^b x^{-1} \, dx = \ln\left(\frac{b}{a}\right).
\]

For \(i \geq 1\), we have
\[
\int_a^b x^{i - 1} \, dx = \frac{b^i - a^i}{i}.
\]

Substituting back into the sum, we obtain:
\[
\int_a^b \frac{(1 - x)^m}{x} \, dx = \binom{m}{0} (-1)^0 \ln\left(\frac{b}{a}\right) + \sum_{i=1}^{m} \binom{m}{i} (-1)^i \cdot \frac{b^i - a^i}{i}.
\]

Since \(\binom{m}{0} = 1\) and \((-1)^0 = 1\), the final result is:
\[
\int_a^b \frac{(1 - x)^m}{x} \, dx = \ln\left(\frac{b}{a}\right) + \sum_{i=1}^{m} \binom{m}{i} (-1)^i \cdot \frac{b^i - a^i}{i}.
\]
\end{proof}
%------------------------------------------------------------------
\begin{lemma}
\label{lem:lemma41}   
\begin{equation}
\int_{\tau}^{\beta} \left(1 - (1 - t^*)^{m-1} \right) \cdot \frac{\tau}{t^*} \, dt^*
= \tau \sum_{i=1}^{m-1} \binom{m-1}{i} (-1)^{i+1} \cdot \frac{\beta^i - \tau^i}{i}.
\end{equation}
\end{lemma}

\begin{proof}
We begin by factoring out the constant \(\tau\) from the integral:
\[
\int_{\tau}^{\beta} \left(1 - (1 - t^*)^{m-1} \right) \cdot \frac{\tau}{t^*} \, dt^*
= \tau \int_{\tau}^{\beta} \left(1 - (1 - t^*)^{m-1} \right) \cdot \frac{1}{t^*} \, dt^*.
\]

Next, split the integral:
\[
\tau \int_{\tau}^{\beta} \left(1 - (1 - t^*)^{m-1} \right) \cdot \frac{1}{t^*} \, dt^*
= \tau \left[ \int_{\tau}^{\beta} \frac{1}{t^*} \, dt^*
- \int_{\tau}^{\beta} \frac{(1 - t^*)^{m-1}}{t^*} \, dt^* \right].
\]

From~\cref{lemma:lm3}, we know:
\[
\int_{\tau}^{\beta} \frac{1}{t^*} \, dt^* = \ln\left( \frac{\beta}{\tau} \right),
\]
and from~\cref{lemma:lm4} with \( m \leftarrow m-1 \), we get:
\[
\int_{\tau}^{\beta} \frac{(1 - t^*)^{m-1}}{t^*} \, dt^*
= \sum_{i=1}^{m-1} \binom{m-1}{i} (-1)^i \cdot \frac{\beta^i - \tau^i}{i}
+ \ln\left( \frac{\beta}{\tau} \right).
\]

Substituting both expressions back:
\begin{align*}
&\tau \left[ \ln\left( \frac{\beta}{\tau} \right)
- \left( \sum_{i=1}^{m-1} \binom{m-1}{i} (-1)^i \cdot \frac{\beta^i - \tau^i}{i}
+ \ln\left( \frac{\beta}{\tau} \right) \right) \right] \\
&= -\tau \sum_{i=1}^{m-1} \binom{m-1}{i} (-1)^i \cdot \frac{\beta^i - \tau^i}{i} \\
&= \tau \sum_{i=1}^{m-1} \binom{m-1}{i} (-1)^{i+1} \cdot \frac{\beta^i - \tau^i}{i}.
\end{align*}
\end{proof}

\begin{lemma}
\label{lem:lemma42}  
\begin{align}
\int_{\beta}^{1}\left(1 - (1-t^*)^k \right) \cdot \frac{\tau}{t^*} 
 dt^* =  \tau \sum_{i=1}^{k} \binom{k}{i} (-1)^{i+1} \cdot \frac{1 - \beta^i}{i}.
\end{align}
\end{lemma}
\begin{proof}
The proof is identical to that of~\cref{lem:lemma41}, except with integration bounds \([\beta, 1]\) instead of \([\tau, \beta]\), and using \(k\) in place of \(m - 1\). The same binomial expansion and application of~\cref{lemma:lm3} and~\cref{lemma:lm4} apply.
\end{proof}

\begin{lemma}
\label{lem:lemma43}  
\begin{equation}
\begin{aligned}
\int_{\beta}^{1} (1 - t^*)^k \left[
\left(1 - (1 - \beta)^{m_2} \right)(1 - \delta)\frac{\tau}{\beta} + 
(1 - \beta)^{m_2}(1 - \gamma)
\right] dt^* \\
= \frac{(1 - \beta)^{k+1}}{k+1} \cdot \left[
\left(1 - (1 - \beta)^{m_2} \right)(1 - \delta)\frac{\tau}{\beta} + 
(1 - \beta)^{m_2}(1 - \gamma)
\right]
\end{aligned}
\end{equation}

\end{lemma}

\begin{proof}
Let us denote the constant coefficient by:
\[
C := \left[
\left(1 - (1 - \beta)^{m_2} \right)(1 - \delta)\frac{\tau}{\beta} + 
(1 - \beta)^{m_2}(1 - \gamma)
\right].
\]

Then the integral becomes:
\[
\int_{\beta}^{1} (1 - t^*)^k \cdot C \, dt^* = C \cdot \int_{\beta}^{1} (1 - t^*)^k \, dt^*.
\]

Now apply~\cref{lemma:lm1} with \( m = k + 1 \), \( a = \beta \), \( b = 1 \):
\[
\int_{\beta}^{1} (1 - t^*)^k \, dt^* = \frac{(1 - \beta)^{k+1}}{k + 1}.
\]

Thus:
\[
\int_{\beta}^{1} (1 - t^*)^k \cdot C \, dt^* = \frac{(1 - \beta)^{k+1}}{k+1} \cdot C,
\]
which proves the lemma.
\end{proof}

%---------------------------------------------------------------
\begin{lemma}
\label{lem:lemma51}
\begin{align}
\int_{\tau}^{\beta} \left[
\left(1 - (1 - t^*)^{m-1} \right) \cdot \frac{\tau}{t^*} + (1 - t^*)^{m-1}
\right] dt^* \notag \\
= \tau \sum_{i=1}^{m-1} \binom{m-1}{i} (-1)^{i+1} \cdot \frac{\beta^i - \tau^i}{i} 
+ \frac{(1 - \tau)^m - (1 - \beta)^m}{m}.
\end{align}
\end{lemma}
\begin{proof}
We split the integral into two parts:
\[
\int_{\tau}^{\beta} \left[
\left(1 - (1 - t^*)^{m-1} \right) \cdot \frac{\tau}{t^*} + (1 - t^*)^{m-1}
\right] dt^*
= I_1 + I_2,
\]
where
\[
I_1 = \int_{\tau}^{\beta} \left(1 - (1 - t^*)^{m-1} \right) \cdot \frac{\tau}{t^*} \, dt^*, \quad
I_2 = \int_{\tau}^{\beta} (1 - t^*)^{m-1} \, dt^*.
\]

From~\cref{lem:lemma41}, we have:
\[
I_1 = \tau \sum_{i=1}^{m-1} \binom{m-1}{i} (-1)^{i+1} \cdot \frac{\beta^i - \tau^i}{i}.
\]

From~\cref{lemma:lm1} with \(f_X(x) = m(1 - x)^{m-1}\), so \( (1 - x)^{m-1} = \frac{1}{m} f_X(x) \), we get:
\[
I_2 = \int_{\tau}^{\beta} (1 - t^*)^{m-1} \, dt^*
= \frac{1}{m} \int_{\tau}^{\beta} f_X(t^*) \, dt^*
= \frac{1}{m} \left[(1 - \tau)^m - (1 - \beta)^m \right].
\]

Combining the two:
\[
I_1 + I_2 = \tau \sum_{i=1}^{m-1} \binom{m-1}{i} (-1)^{i+1} \cdot \frac{\beta^i - \tau^i}{i} 
+ \frac{(1 - \tau)^m - (1 - \beta)^m}{m}.
\]
\end{proof}

\begin{lemma}
\label{lem:lemma52}  
\begin{align}
\int_{\beta}^{1} \left(1 - (1 - t^*)^k \right) \cdot \frac{\tau}{t^*} \cdot \left(1 - (1 - \beta)^{m - 1} \right) dt^* \notag \\
= \tau \sum_{i=1}^{k} \binom{k}{i} (-1)^{i+1} \cdot \frac{1 - \beta^i}{i} \cdot \left(1 - (1 - \beta)^{m - 1} \right).
\end{align}
\end{lemma}

\begin{proof}
Let 
\[
C := \left(1 - (1 - \beta)^{m - 1} \right),
\]
which is constant with respect to \(t^*\). Then the integral becomes:
\[
\int_{\beta}^{1} \left(1 - (1 - t^*)^k \right) \cdot \frac{\tau}{t^*} \cdot C \, dt^*
= C \cdot \int_{\beta}^{1} \left(1 - (1 - t^*)^k \right) \cdot \frac{\tau}{t^*} \, dt^*.
\]

Using~\cref{lem:lemma42}, we substitute:
\[
\int_{\beta}^{1} \left(1 - (1 - t^*)^k \right) \cdot \frac{\tau}{t^*} \, dt^* 
= \tau \sum_{i=1}^{k} \binom{k}{i} (-1)^{i+1} \cdot \frac{1 - \beta^i}{i}.
\]

Multiplying by \(C\), we get:
\[
\int_{\beta}^{1} \left(1 - (1 - t^*)^k \right) \cdot \frac{\tau}{t^*} \cdot C \, dt^* 
= \tau \sum_{i=1}^{k} \binom{k}{i} (-1)^{i+1} \cdot \frac{1 - \beta^i}{i} \cdot \left(1 - (1 - \beta)^{m - 1} \right).
\]
\end{proof}

\begin{lemma}
\label{lem:lemma53}  
\begin{align}
\int_{\beta}^{1} (1 - t^*)^k \cdot \left[ \left(1 - (1 - \beta)^{m_2} \right)(1 - \delta)\frac{\tau}{\beta} \right] dt^* \notag \\
= \frac{(1 - \beta)^{k+1}}{k+1} \cdot \left[ \left(1 - (1 - \beta)^{m_2} \right)(1 - \delta)\frac{\tau}{\beta} \right].
\end{align}
\end{lemma}

\begin{proof}
Let us denote the constant factor by:
\[
C := \left(1 - (1 - \beta)^{m_2} \right)(1 - \delta)\frac{\tau}{\beta}.
\]

Then the integral simplifies to:
\[
\int_{\beta}^{1} (1 - t^*)^k \cdot C \, dt^* = C \cdot \int_{\beta}^{1} (1 - t^*)^k \, dt^*.
\]

Applying the standard power integral:
\[
\int_{\beta}^{1} (1 - t^*)^k \, dt^* = \frac{(1 - \beta)^{k+1}}{k + 1},
\]
we obtain:
\[
\int_{\beta}^{1} (1 - t^*)^k \cdot C \, dt^* = \frac{(1 - \beta)^{k+1}}{k + 1} \cdot C.
\]
\end{proof}

\begin{lemma}
\label{lem:lemmaX41}  
\begin{equation}
\int_0^{\tau}\left(\int_{\tau}^{1}\frac{\tau}{t^*}\left(1 - (1 - t^*)^k\right) + (1 - t^*)^k \,dt^*\right)d\beta
= \tau\left[-\tau\sum_{i=1}^{k} \binom{k}{i} (-1)^i \frac{1 - \tau^i}{i} + \frac{(1 - \tau)^{k+1}}{k+1}\right]
\end{equation}
\end{lemma}

\begin{proof}
We begin by observing that the integrand does not depend on \(\beta\), so the outer integral is simply a multiplication by \(\tau\):
\[
\int_0^{\tau} \left( \int_{\tau}^{1} \left[ \frac{\tau}{t^*}(1 - (1 - t^*)^k) + (1 - t^*)^k \right] dt^* \right) d\beta
= \tau \cdot \int_{\tau}^{1} \left[ \frac{\tau}{t^*}(1 - (1 - t^*)^k) + (1 - t^*)^k \right] dt^*.
\]

We now compute the inner integral, splitting it as:
\[
\int_{\tau}^{1} \left[ \frac{\tau}{t^*}(1 - (1 - t^*)^k) \right] dt^*
+ \int_{\tau}^{1} (1 - t^*)^k dt^*.
\]

The first term is handled by~\cref{lem:lemma42}, applied on \([\tau, 1]\):
\[
\int_{\tau}^{1} \frac{\tau}{t^*}(1 - (1 - t^*)^k) \, dt^*
= \tau \sum_{i=1}^k \binom{k}{i} (-1)^{i+1} \cdot \frac{1 - \tau^i}{i}
= -\tau \sum_{i=1}^k \binom{k}{i} (-1)^i \cdot \frac{1 - \tau^i}{i}.
\]

The second term is a standard power integral:
\[
\int_{\tau}^{1} (1 - t^*)^k dt^* = \frac{(1 - \tau)^{k+1}}{k + 1}.
\]

Putting both together:
\[
\tau \cdot \left( -\tau \sum_{i=1}^k \binom{k}{i} (-1)^i \cdot \frac{1 - \tau^i}{i}
+ \frac{(1 - \tau)^{k+1}}{k + 1} \right),
\]
which completes the proof.
\end{proof}

\begin{lemma}
\label{lem:lemmaX42}  
 \begin{equation}
\begin{aligned}
 \int_{\tau}^{1} C_4(m) \, d\beta   =
 & \tau \sum_{i=1}^{m-1} \binom{m-1}{i} (-1)^{i+1} \left( \frac{1 - \tau^{i+1}}{i(i+1)} - \frac{\tau^i (1 - \tau)}{i} \right) \\
& + \tau \sum_{i=1}^{k} \binom{k}{i} (-1)^{i+1} \left( \frac{1 - \tau}{i} - \frac{1 - \tau^{i+1}}{i(i+1)} \right) \\
& + \frac{(1 - \delta)\tau}{k+1} \left[ 
\sum_{i=1}^{k+1} \binom{k+1}{i} (-1)^i \frac{1 - \tau^i}{i}
- \sum_{i=1}^{k+1+m_2} \binom{k+1+m_2}{i} (-1)^i \frac{1 - \tau^i}{i}
\right] \\
& + \frac{(1 - \gamma)}{k+1} \cdot \frac{(1 - \tau)^{k+2 + m_2}}{k+2 + m_2}
\end{aligned}  
 \end{equation}
\end{lemma}

\begin{proof}
We integrate each term of \( C_4(m) \) over the interval \( \beta \in [\tau, 1] \). From the definition:
\[
C_4(m) = 
\tau \sum_{i=1}^{m-1} \binom{m-1}{i} (-1)^{i+1} \cdot \frac{\beta^i - \tau^i}{i}
+ \tau \sum_{i=1}^{k} \binom{k}{i} (-1)^{i+1} \cdot \frac{1 - \beta^i}{i}
+ \frac{(1 - \beta)^{k+1}}{k + 1} \cdot \left[
A(\beta) + B(\beta)
\right],
\]
where
\[
A(\beta) := \left(1 - (1 - \beta)^{m_2} \right)(1 - \delta)\frac{\tau}{\beta}, \qquad
B(\beta) := (1 - \beta)^{m_2}(1 - \gamma).
\]

\textbf{First term:}
\[
\int_{\tau}^{1} \tau \sum_{i=1}^{m-1} \binom{m-1}{i} (-1)^{i+1} \cdot \frac{\beta^i - \tau^i}{i} \, d\beta
= \tau \sum_{i=1}^{m-1} \binom{m-1}{i} (-1)^{i+1}
\left( \frac{1 - \tau^{i+1}}{i(i+1)} - \frac{\tau^i(1 - \tau)}{i} \right).
\]

\textbf{Second term:}
\[
\int_{\tau}^{1} \tau \sum_{i=1}^{k} \binom{k}{i} (-1)^{i+1} \cdot \frac{1 - \beta^i}{i} \, d\beta
= \tau \sum_{i=1}^{k} \binom{k}{i} (-1)^{i+1}
\left( \frac{1 - \tau}{i} - \frac{1 - \tau^{i+1}}{i(i+1)} \right).
\]

\textbf{Third term (part A):}
Using~\cref{lem:lemma42} applied to the function \( A(\beta) \), and expanding:
\[
\int_{\tau}^{1} \frac{(1 - \beta)^{k+1}}{k+1} \cdot A(\beta) \, d\beta
= \frac{(1 - \delta)\tau}{k+1} \left[
\sum_{i=1}^{k+1} \binom{k+1}{i} (-1)^i \cdot \frac{1 - \tau^i}{i}
- \sum_{i=1}^{k+1+m_2} \binom{k+1+m_2}{i} (-1)^i \cdot \frac{1 - \tau^i}{i}
\right].
\]

\textbf{Fourth term (part B):}
\[
\int_{\tau}^{1} \frac{(1 - \beta)^{k+1 + m_2}}{k+1} (1 - \gamma) \, d\beta
= \frac{(1 - \gamma)}{k+1} \cdot \frac{(1 - \tau)^{k+2 + m_2}}{k+2 + m_2}.
\]

Adding all terms completes the proof.
\end{proof}

\begin{lemma}
\label{lem:lemmaX51}
\begin{equation}
\begin{aligned}
\int_0^{\tau} \left( \int_{\tau}^{1}
\frac{\tau}{t^*} \left(1 - (1 - t^*)^k\right) \left(1 - (1 - \beta)^m\right)
+ (1 - t^*)^k \left(1 - (1 - \beta)^{m_2}\right) \, dt^* \right) d\beta \\
= - \tau \left( \sum_{i=1}^{k} \binom{k}{i} (-1)^i \frac{1 - \tau^i}{i} \right)
\left( \tau - \frac{1 - (1 - \tau)^{m+1}}{m+1} \right) \\
\quad + \frac{(1 - \tau)^{k+1}}{k + 1} \left( \tau - \frac{1 - (1 - \tau)^{m_2+1}}{m_2+1} \right)
\end{aligned}
\end{equation}
\end{lemma}

\begin{proof}
We split the double integral into two parts:
\[
\int_0^{\tau} \left( \int_{\tau}^{1}
\frac{\tau}{t^*} \left(1 - (1 - t^*)^k\right) \left(1 - (1 - \beta)^m\right) \, dt^* \right) d\beta
+
\int_0^{\tau} \left( \int_{\tau}^{1}
(1 - t^*)^k \left(1 - (1 - \beta)^{m_2}\right) \, dt^* \right) d\beta.
\]

For the first term, the inner integral is
\[
\int_{\tau}^{1} \frac{\tau}{t^*} \left(1 - (1 - t^*)^k\right) \, dt^*
= \tau \sum_{i=1}^{k} \binom{k}{i} (-1)^{i+1} \cdot \frac{1 - \tau^i}{i}
= -\tau \sum_{i=1}^{k} \binom{k}{i} (-1)^i \cdot \frac{1 - \tau^i}{i}.
\]
The outer integral is
\[
\int_0^{\tau} \left(1 - (1 - \beta)^m \right) \, d\beta
= \tau - \int_0^{\tau} (1 - \beta)^m \, d\beta
= \tau - \frac{1 - (1 - \tau)^{m+1}}{m+1}.
\]
Multiplying the two gives:
\[
- \tau \left( \sum_{i=1}^{k} \binom{k}{i} (-1)^i \cdot \frac{1 - \tau^i}{i} \right)
\left( \tau - \frac{1 - (1 - \tau)^{m+1}}{m+1} \right).
\]

For the second term, the inner integral is
\[
\int_{\tau}^{1} (1 - t^*)^k \, dt^* = \frac{(1 - \tau)^{k+1}}{k + 1},
\]
and the outer integral is
\[
\int_0^{\tau} \left(1 - (1 - \beta)^{m_2} \right) \, d\beta
= \tau - \frac{1 - (1 - \tau)^{m_2 + 1}}{m_2 + 1}.
\]
Multiplying them yields:
\[
\frac{(1 - \tau)^{k+1}}{k + 1} \cdot \left( \tau - \frac{1 - (1 - \tau)^{m_2 + 1}}{m_2 + 1} \right).
\]

Adding both parts completes the proof.
\end{proof}

\begin{lemma}
\label{lem:lemmaX52}
Let \( C_5(m) \) be defined as in the display. Then,
\begin{equation}
\begin{aligned}
 \int_{\tau}^{1} C_5(m) \, d\beta
&\geq \tau \sum_{i=1}^{m-1} \binom{m-1}{i} (-1)^{i+1}
\left( \frac{1 - \tau^{i+1}}{i(i+1)} - \frac{\tau^i (1 - \tau)}{i} \right) \\
&\quad + \frac{(1 - \tau)^m (1 - \tau)}{m} - \frac{(1 - \tau)^{m+1}}{m(m+1)} \\
&\quad + \tau \sum_{i=1}^{k} \binom{k}{i} (-1)^{i+1} \frac{1}{i}
\left[
(1 - \tau) - \frac{1 - \tau^{i+1}}{i+1} - \frac{(1 - \tau)^m}{m}
+ \frac{i!\,(m-1)!}{(i + m)!} - \int_0^{\tau} \beta^i (1 - \beta)^{m - 1} d\beta
\right] \\
&\quad + \frac{(1 - \delta)\tau}{k+1} \left[
\sum_{i=1}^{k+1} \binom{k+1}{i} (-1)^i \frac{1 - \tau^i}{i}
- \sum_{i=1}^{k+1 + m_2} \binom{k+1 + m_2}{i} (-1)^i \frac{1 - \tau^i}{i}
\right]
\end{aligned}
\end{equation}
\end{lemma}

\begin{proof}
We analyze each term of the integrand \( C_5(m) \), given by:
\[
C_5(m) =
\tau \sum_{i=1}^{m-1} \binom{m-1}{i} (-1)^{i+1} \cdot \frac{\beta^i - \tau^i}{i}
+ \frac{(1 - \tau)^m - (1 - \beta)^m}{m}
+ \tau \sum_{i=1}^{k} \binom{k}{i} (-1)^{i+1} \cdot \frac{1 - \beta^i}{i} \cdot \left(1 - (1 - \beta)^{m-1} \right) \\
+ \frac{(1 - \beta)^{k+1}}{k+1} \cdot \left[
(1 - (1 - \beta)^{m_2}) (1 - \delta)\frac{\tau}{\beta}
\right].
\]

We integrate each term over \( \beta \in [\tau, 1] \):

First term:
\[
\tau \sum_{i=1}^{m-1} \binom{m-1}{i} (-1)^{i+1} \int_{\tau}^{1} \frac{\beta^i - \tau^i}{i} d\beta
= \tau \sum_{i=1}^{m-1} \binom{m-1}{i} (-1)^{i+1} \left( \frac{1 - \tau^{i+1}}{i(i+1)} - \frac{\tau^i (1 - \tau)}{i} \right).
\]

Second term:
\[
\int_{\tau}^{1} \frac{(1 - \tau)^m - (1 - \beta)^m}{m} \, d\beta
= \frac{(1 - \tau)^m (1 - \tau)}{m} - \frac{(1 - \tau)^{m+1}}{m(m+1)}.
\]

Third term:
This requires expanding both parts:
\[
\int_{\tau}^{1} \frac{1 - \beta^i}{i} \cdot \left(1 - (1 - \beta)^{m-1} \right) d\beta
= \left[ \frac{1 - \tau}{i} - \frac{1 - \tau^{i+1}}{i(i+1)} \right]
- \left[ \frac{(1 - \tau)^m}{m i} - \frac{i!(m-1)!}{(i + m)!} + \int_0^{\tau} \beta^i (1 - \beta)^{m - 1} d\beta \right].
\]
Multiplying by \(\tau\binom{k}{i}(-1)^{i+1}\) and summing gives the third block.

Fourth term:
The integral of the last term is given directly by~\cref{lem:lemma42} with a prefactor \((1 - \delta)\tau/(k+1)\) and using binomial identity expansions:
\[
\int_{\tau}^{1} \frac{(1 - \beta)^{k+1}}{\beta} \, d\beta
= \sum_{i=1}^{k+1} \binom{k+1}{i} (-1)^i \frac{1 - \tau^i}{i}, \quad
\int_{\tau}^{1} \frac{(1 - \beta)^{k+1 + m_2}}{\beta} \, d\beta
= \sum_{i=1}^{k+1 + m_2} \binom{k+1 + m_2}{i} (-1)^i \frac{1 - \tau^i}{i}.
\]
Subtracting gives the fourth term.

Combining all four completes the proof.
\end{proof}

%\begin{lemma}
%\label{lem:lemma41}  
%\end{lemma}

\end{document}